\documentclass[12pt]{article}

\usepackage{booktabs}
\usepackage{times}
\usepackage{graphicx}
\usepackage{colortbl}
\usepackage{tabularx}
\usepackage{amsmath,amsthm,amsfonts}
\usepackage{amssymb}
\usepackage{rotating}
\usepackage{hyperref}

\DeclareMathAlphabet\mathbfcal{OMS}{cmsy}{b}{n}

\newcommand{\yes}{Yes}
\newcommand{\no}{No}
\newcommand{\yesins}{Yes-instance}
\newcommand{\noins}{No-instance}

\newcommand{\fpt}{\textsf{FPT}}
\newcommand{\np}{\textsf{NP}}
\newcommand{\nph}{{\np}-hard}
\newcommand{\nphshort}{{{\np}-h}}
\newcommand{\nphns}{{{\np}-hardness}}

\newcommand{\poly}{{\textsf{P}}}

\newcommand{\bigo}{O}
\newcommand{\bigos}{O^*}
\newcommand{\discset}{J}
\newcommand{\setmid}{:}
\newcommand{\abs}[1]{\vert#1\vert}
\newcommand{\xc}{\mathcal{H}}  
\newcommand{\xce}{H}  
\newcommand{\xs}{A} 
\newcommand{\xse}{a} 
\newcommand{\xsize}{\kappa} 
\newcommand{\wah}{{\textsf{W[1]}}-hard}
\newcommand{\wahshort}{{\textsf{W[1]}}-h}
\newcommand{\wahns}{{\textsf{W[1]}}-hardness}
\newcommand{\prob}[1]{{\sc #1}}
\newcommand{\abmr}{ABM rule}
\newcommand{\abmrs}{ABM rules}
\newcommand{\pappadd}[1]{\prob{DesAppAddBrib-#1}}
\newcommand{\pappdel}[1]{\prob{DesAppDelBrib-#1}}
\newcommand{\pvc}[1]{\prob{$r$-DesVoteChgBrib}-#1}
\newcommand{\pvac}[1]{\prob{$r$-DesVoteAddChgBrib}-#1}
\newcommand{\pvdc}[1]{\prob{$r$-DesVoteDelChgBrib}-#1}

\newcommand{\vset}{U}  
\newcommand{\eset}{A}  

\newcommand{\hamdis}[2]{\textsf{Ham}(#1, #2)}

\newcommand{\edge}[2]{\{#1,#2\}}
\newcommand{\hide}[1]{}
\newcommand{\p}{p}

\renewcommand{\iff}{if and only if}
\newcommand{\thm}{Thm.~}
\newcommand{\thms}{Thms.~}

\newcommand{\av}{{\text{AV}}}
\newcommand{\sav}{{\text{SAV}}}

\newcommand{\nsav}{{\text{NSAV}}}
\newcommand{\vccav}{{\text{CCAV}}}
\newcommand{\pav}{{\text{PAV}}}

\newcommand{\n}{n} 
\newcommand{\m}{m}  
\newcommand{\they}{they}
\newcommand{\their}{their}
\newcommand{\EP}[3]{
\begin{center}
{\small
\begin{tabularx}{0.98\columnwidth}{ll}
\toprule
\multicolumn{2}{c}{{#1}} \\
\midrule
{\bf Input:}   & \parbox[t]{0.85\columnwidth}{#2\vspace*{1mm}}  \\
{\bf Question:}& \parbox[t]{0.85\columnwidth}{#3\vspace*{.5mm}} \\
\bottomrule
\end{tabularx}
}
\end{center}
}

 \usepackage{geometry}
\newcommand{\score}[4]{{\textsf{sc}}^{#3}_{#4}(#1,#2)} 
\newcommand{\scoreE}[3]{\textsf{sc}^{#2}_{#3}(#1)} 
\newcommand{\dd}{d}

\let\shortcite\cite

\newtheorem{theorem}{Theorem}
\newtheorem{corollary}{Corollary}
\newtheorem{claim}{Claim}
\newtheorem{lemma}{Lemma}
\newtheorem{reductionrule}{Reduction Rule}

\begin{document}

\title{On the Complexity of Destructive Bribery in Approval-Based Multiwinner Voting\thanks{A preliminary version of this paper appeared in the Proceedings of the 19th International Conference on Autonomous Agents and Multiagent Systems (AAMAS 2020)~\protect\cite{DBLP:conf/atal/000120}.}
}

\author{Yongjie Yang}
\date{\small{Chair of Economic Theory, Saarland University, Saarbr\"{u}cken 66123, Germany\\ Email: yyongjiecs@gmail.com}}

\maketitle

\begin{abstract}
A variety of constructive manipulation, control, and bribery problems for approval-based multiwinner voting have been extensively studied recently. However, their destructive counterparts seem to be less explored. This paper investigates the complexity of several destructive bribery problems under five prestigious approval-based multiwinner voting rules---approval voting, satisfaction approval voting, net-satisfaction approval voting, Chamberlin-Courant approval voting, and proportional approval voting. Broadly, these problems are to determine if a number of given candidates can be excluded from any winning committees by performing a limited number of modification operations. We offer a complete landscape of the complexity of the problems. For {\nph} problems, we study their parameterized complexity with respect to meaningful parameters.
\smallskip

\end{abstract}


\section{Introduction}
After more than two decades of extensive study on the complexity of single-winner voting problems, the computational social choice community has recently shifted its primary focus to multiwinner voting, given its generality and broad applications. 
In particular, many variants of manipulation, control, and bribery problems for approval-based multiwinner voting rules ({\abmrs} for short) have been studied from a complexity point of view (see e.g.,~\cite{DBLP:conf/atal/AzizGGMMW15,DBLP:conf/atal/FaliszewskiST17,DBLP:conf/ijcai/Yang19,DBLP:conf/atal/ZhouYG19}). Existing works in this line of research predominantly concern the constructive model of these problems, which models scenarios where a strategic agent attempts to elevate a single distinguished candidate to winner status, or make a committee a winning committee. However, the destructive counterparts of these problems have not been adequately studied in the literature so far. This paper studies the complexity and the parameterized complexity of several destructive bribery problems for {\abmrs}. These problems are designed to capture scenarios where an election attacker (or briber) aims to prevent multiple distinguished candidates from winning by making changes to the votes (e.g., by bribing some voters to alter their votes) under certain budget constraints. The attacker's motivation may stem from these distinguished candidates being rivals (e.g., having completely different political views from the attacker), or the attacker attempting to make them lose to increase the winning chance of her preferred candidates. 

We consider five  bribery operations categorized into two classes: {\it{atomic operations}} and {\it{vote-level operations}}. Approval addition (AppAdd) and approval deletion (AppDel) are atomic operations, where each single AppAdd/AppDel means to add/delete one candidate into/from the set of approved candidates in some vote. Vote-level change (VoteChg), vote-level addition change (VoteAddChg), and vote-level deletion change (VoteDelChg) are vote-level operations, where each single operation involves changing a vote in any possible way, changing a vote by adding some candidates to the set of approved candidates, and changing a vote by deleting some candidates from the set of approved candidates, respectively. Each bribery problem is associated with an operation type, and the attacker can perform at most a given number of single operations of the same type. For vote-level operation-based problems, we introduce a distance bound~$r$ and assume that each vote can only be changed into another one with a Hamming distance at most~$r$ from the original vote. We call bribery problems with this parameter {\emph{distance-bounded bribery problems}}. This parameter models scenarios where voters prefer not to deviate too much from their true preferences. This parameter broadens the scope our study, since when~$r$ equals the number of candidates, the impact of this distance bound completely fades out. 


We investigate these problems under five widely-studied {\abmrs}, namely, approval voting ({\av}), satisfaction approval voting ({\sav}), net-satisfaction approval voting ({\nsav}), Chamberlin-Courant approval voting ({\vccav}), and proportional approval voting ({\pav}). We determine the complexity of all problems considered in the paper. Notably, many of our {\nphns} results hold even in very special cases. For {\nph} problems, we explore how various meaningful parameters influence their parameterized complexity. The considered parameters encompass the number of candidates, the number of votes, the size of the winning committee, the number of distinguished candidates, the number of operations the strategic agent is allowed to perform, and the distance bound discussed above. 

\subsection{Related Works}

Our work is clearly related to the pioneering studies of Bartholdi and Orlin~\cite{BartholdiJames1991SocialChoiceWelfareSTV}, as well as Bartholdi, Tovey, and Trick.~\cite{BARTHOLDI89,Bartholdi92howhard}, which examined numerous strategic single-winner voting problems from a complexity perspective. Their motivation stemmed from the idea that computational complexity can serve as a barrier against strategic behavior.\footnote{It should be pointed out that there exist experimental studies demonstrating that many computationally hard voting problems can be solved efficiently for practical elections (see, e.g.,~\cite{DBLP:conf/sagt/BoehmerFJK23,DBLP:journals/jcss/ErdelyiFRS15a,DBLP:conf/aaai/LoreggiaNRVW14}).} Since their seminal work, research on the complexity of single-winner voting problems---particularly strategic problems in both constructive and destructive models---has driven much of the progress in computational social choice. However, the study of the complexity of multiwinner voting problems had lagged behind, with only a few related papers published. Among these studies, the following works are most relevant to ours. Meir~et~al.~\shortcite{DBLP:journals/jair/MeirPRZ08} explored both constructive and destructive manipulation and control, but mainly for ranking-based multiwinner voting rules. Faliszewski, Skowron, and Talmon~\shortcite{DBLP:conf/atal/FaliszewskiST17} analyzed various constructive bribery problems for {\abmrs}, including the operations AppAdd and AppDel. Bredereck~et~al.~\shortcite{DBLP:conf/aaai/BredereckFNT16} studied constructive shift bribery for both ranking-based multiwinner voting rules and {\abmrs}.  
However, in these problems, there is only a single distinguished candidate whom the election attackers seek to include in the winning committee. In contrast, our study considers multiple distinguished candidates whom the attackers aim to exclude from all winning committees. Yang~\shortcite{DBLP:conf/ijcai/Yang19} delved into the complexity of manipulation and control problems for {\abmrs} with multiple distinguished candidates but  focused exclusively on the constructive model. Destructive strategic voting problems with multiple distinguished candidates have been explored in the context of single-winner voting problems~\cite{DBLP:conf/atal/YangW17} and group identification~\cite{DBLP:journals/aamas/ErdelyiRY20,DBLP:journals/aamas/YangD18}.  Aziz~et~al.~\shortcite{DBLP:conf/atal/AzizGGMMW15} also studied a constructive manipulation problem in which a given set of candidates must form the winning committee. 

Our study of the distance-bounded bribery problems are related to~\cite{DBLP:conf/aaai/BaumeisterHR19,DBLP:conf/atal/Dey19,DBLP:journals/tcs/YangSG19}, where bribery problems with distance restrictions have been studied for ranking-based single-winner voting rules. 

The destructive bribery problems studied in this paper are connected to the concept of robustness of multiwinner voting rules, which focuses on the minimum changes required to alter the winning committees. Notably, Gawron, and Faliszewski~\shortcite{DBLP:conf/aldt/GawronF19} recently studied the complexity of determining the number of operations needed to change the set of winning committees in approval-based voting. They considered the addition, deletion, and replacement operations. Our problems have different objectives. 

In comparison with the conference version~\cite{DBLP:conf/atal/000120}, the current iteration contains all previously missing proofs. 
We note that since the publication of~\cite{DBLP:conf/atal/000120}, significant advancements have been made in this line of research. Specifically, Kusek~et~al.~\cite{DBLP:conf/atal/KusekBF0K23} explored the complexity of constructive bribery problems for the rule AV restricted to the candidate interval domain and the voter interval domain. They also provided a glimpse into the complexity of the destructive counterparts within the same restricted setting. 
Another notable contribution comes from the work of Faliszewski, Gawron, and Kusek~\cite{DBLP:conf/eumas/FaliszewskiGK22}, where they extended the scope of the study on robustness of {\abmrs} to encompass several greedy approval rules.

\subsection{Organization} 
The remainder of the paper is organized as follows. In Section~\ref{sec:preliminaries}, we present definitions of crucial concepts pertinent to our study. Moving on to Section~\ref{sec-ccav-pav}, we delve into the destructive bribery problems for {\vccav} and {\pav}. 
Subsequently, the following section, consisting of five subsections, is dedicated to exploring the complexity of the five destructive bribery problems for {\av}, {\sav}, and {\nsav}. Following this, in Section~\ref{sec-many-fpts}, we study {\fpt}-algorithms for three parameters---the number of candidates, the number of votes, and the number of distinguished candidates. Our study concludes in Section~\ref{sec-conclusion}, where we summarize our findings and outline several intriguing avenues for future research.

\section{Preliminaries}
\label{sec:preliminaries}
\subsection{Multiwinner Voting}
In approval-based voting, each voter is asked to report a subset of candidates that they approve of. Formally, an approval-based election is represented as a tuple $(C, V)$, where $C$ is a set of candidates and~$V$ is a multiset of votes.\footnote{The multiset notation accounts for multiple voters casting identical approval sets.} Each vote, cast by a voter, is defined as a subset of candidates consisting of all candidates approved by that voter. An empty vote is one without any candidates. A subset of candidates is referred to as a committee, while a subset consisting of exactly~$k$ candidates, for some integer~$k$, is termed a $k$-committee. A multiwinner voting rule~$f$ assigns to each election $(C, V)$ and an integer $k \leq |C|$ a collection of $k$-committees—--the winning $k$-committees of~$f$ at $(C, V)$. For $k = 1$, we refer to the candidate in each $1$-winning committee as a winner, and if the rule returns only a singleton $\{\{c\}\}$, where $c\in C$, we call~$c$ the unique winner of $f$ at $(C, V)$.

This paper focuses on the rules {\av}, {\sav}, {\nsav}, {\vccav}, and {\pav}.
In these rules, each vote assigns a certain score to each committee, and winning $k$-committees are those with the maximum total score received from all votes. These rules differ only in how the scores are defined. A summary of committee score definitions across different rules is given in Table~\ref{tab-rules}.
\begin{table}
\caption{A summary of committee score definitions across different rules}
\centering
{
\begin{tabular}{ll}\toprule
rules & score of $w\subseteq C$ in an election $(C, V)$\\ \midrule

{\av} & $\sum_{v\in V} \abs{v\cap w}$\\

{\sav} & $\sum_{v\in V, v\neq \emptyset} \frac{\abs{v\cap w}}{\abs{v}}$ \\

{\nsav} & $\sum_{v\in V, v\neq \emptyset} \frac{\abs{v\cap w}}{\abs{v}}-\sum_{v\in V, v\neq C}\frac{\abs{w\setminus v}}{|C|-\abs{v}}$\\

{\vccav} & $\abs{\{v\in V\setmid v\cap w\neq \emptyset\}}$\\

{\pav} & $\sum_{v\in V, v\cap w\neq \emptyset}\sum_{i=1}^{\abs{v\cap w}}\frac{1}{i}$ \\ \bottomrule
\end{tabular}
}
\label{tab-rules}
\end{table}

In {\av}, each voter gives~$1$ point to every candidate {\they} approve. In {\sav}, each voter has a fixed~$1$ point which is equally distributed among {\their} approved candidates. {\nsav} takes a step further by allowing voters to express dissatisfaction with their disapproved candidates. Particularly, in addition to the fixed~$1$ point equally distributed among {\their} approved candidates, similar to {\sav}, each voter also equally distributes~$-1$ point among all {\their} disapproved candidates. The {\av}/{\sav}/{\nsav} score of a committee is the sum of the total scores of its members. {\sav} and {\nsav} were respectively proposed by Brams and Kilgour~\shortcite{Bram2014Kilgour}, and by Kilgour and Marshall~\shortcite{Kilgour2014Marshall}. 

In {\vccav}, it is presumed that voters only care about whether a committee contains at least one of {\their} approved candidates. A voter is satisfied with a committee if at least one of {\their} approved candidates is included in the committee. This rule selects the $k$-committees satisfying the maximum number of voters. {\vccav} is a special case of a class of rules studied in~\cite{ChamberlinC1983APSR10.2307/1957270} and was suggested by Thiele~\shortcite{Thiele1985}. 

In {\pav}, each committee~$w$ receives $1+\frac{1}{2}+\cdots+\frac{1}{\abs{v\cap w}}$ points from each vote~$v$ such that $v\cap w\neq\emptyset$. {\pav} was first mentioned in the work of Thiele~\shortcite{Thiele1985}. 

It is noteworthy that calculating a winning $k$-committee is {\nph} for {\vccav} and {\pav}, whereas it transitions to a polynomial-time solvable problem for {\av}, {\sav}, and {\nsav}~\cite{DBLP:conf/atal/AzizGGMMW15,DBLP:journals/jair/BetzlerSU13}.

For each $f\in \{\av, \sav, \nsav, \vccav, \pav\}$, an election $E=(C, V)$, a committee~$w\subseteq C$, and a submultiset  $V'\subseteq V$, 
let $\score{V'}{w}{E}{f}$ denote the~$f$ score of~$w$ received from all votes in~$V'$ in the election~$E$. For notational simplicity, we use $\scoreE{w}{E}{f}$ to denote $\score{V}{w}{E}{f}$, the~$f$ score of~$w$ in~$E$. In addition, for~$w$ being a singleton~$\{c\}$, we simply write~$c$ for~$\{c\}$. 


\subsection{The Destructive Bribery Problems}
We explore the five destructive bribery problems characterized by five modification operations, including two atomic operations and three vote-level change operations.
The two atomic operations are defined as follows.

\begin{description}
\item[Approval addition (AppAdd)] A single AppAdd operation on a vote $v\in V$ where $v\neq C$ means the extension of~$v$ by adding exactly one candidate from $C\setminus v$ into~$v$.

\item[Approval deletion (AppDel)] A single AppDel operation on a vote $v\in V$ where $v\neq \emptyset$ entails the removal of one candidate from~$v$.
\end{description}

Let~$f$ be an {\abmr}. Let~$X$ be an atomic operation defined above.

\EP{\prob{Destructive~$X$ Bribery} for~$f$ (\prob{Des$X$Brib-$f$})}
{An election $(C, V)$, a nonempty subset $J\subseteq C$ of distinguished candidates, and two nonnegative integers $k\leq |C|$ and~$\ell$.}
{Is it possible to perform at most~$\ell$ many~$X$ operations on the votes in~$V$ so that none of the candidates in~$J$ is in any winning $k$-committees of the resulting election 
 under the rule~$f$?}

For a {\yesins} $((C, V), J, \ell, k)$ of \prob{Des$X$Brib-$f$}, a feasible solution refers to a sequence of at most~$\ell$ many $X$ operations applied to the votes in $V$, ensuring that none of the candidates in $J$ is contained in any winning $k$-committees of the resulting election. 

Unlike atomic operations, each vote-level change operation modifies a single vote in a specific way.
\begin{description}
\item[Vote-level change (VoteChg)] A single VoteChg operation on a vote~$v$ modifies~$v$ into another vote, which can be any subset of candidates.

\item[Vote-level addition change (VoteAddChg)] A single VoteAddChg operation on a vote~$v$, where $v\neq C$, adds one or more candidates from $C\setminus v$ to~$v$.

\item[Vote-level deletion change (VoteDelChg)] A single VoteDelChg operation on a vote~$v$, where $v\neq\emptyset$, removes one or more candidates from~$v$.
\end{description}

Clearly, each vote-level operation on a vote is equivalent to a sequence of atomic operations applied to the same vote.

The Hamming distance between two votes $v\subseteq C$ and $v'\subseteq C$ is defined as
\[\hamdis{v}{v'}=\abs{v\setminus v'}+\abs{v'\setminus v}.\] 

For a vote-level operation~$Y$ defined above, we study the following distance-bounded bribery problem.

\EP{\prob{$r$-Bounded Destructive~$Y$ Bribery} for~$f$ (\prob{$r$-Des$Y$Brib-$f$})}
{An election $(C, V)$, a nonempty subset $J\subseteq C$ of distinguished candidates, and three nonnegative integers $k\leq |C|$, $\ell\leq |V|$, and~$r$.}
{Is there a subset $V'\subseteq V$ of at most~$\ell$ votes such that we can execute a single~$Y$ operation on each vote in~$V'$ in a way that ensures the Hamming distance between the modified vote and the original vote is at most~$r$ and, moreover, after performing these~$\abs{V'}$ operations, none of the candidates in~$J$ is included in any winning $k$-committees under the rule~$f$?}

For a {\yesins} $I=((C, V), J, \ell, k)$ of \prob{$r$-Des$Y$Brib-$f$}, a feasible solution of~$I$ refers to a submultiset $V'\subseteq V$ of at most~$\ell$ votes and a sequence of~$Y$ operations applied to votes from~$V'$ so that none of~$J$ is contained in any winning $k$-committees of the election after the operations.

\medskip

\noindent{\bf{Remark}.} The winning $k$-committees of every election $(C, V)$, where $V=\emptyset$ or all votes in $V$ are empty votes, are exactly all the $k$-committees of~$C$ under all the five rules studied in this paper. In this case, both {\prob{Des$X$Brib-$f$}} and {\prob{$r$-Des$Y$Brib-$f$}} can be solved trivially for all atomic operations~$X$ and all vote-level operations~$Y$. For the rest of this study, we assume that the election in each instance of the problems contains at least one nonempty vote.

\begin{sidewaystable}
\caption{A summary of the (parameterized) complexity of destructive bribery for {\abmrs}. Here,  ``{\poly}'' means polynomial-time solvable,  ``{\nphshort}'' is shorthand for  ``{\nph}'', and  ``{\wahshort}'' is shorthand for  ``{\wah}''. Parameters for which a parameterized complexity result holds are shown as superscripts. Brackets indicate that the corresponding results remain valid even when the conditions inside the brackets are met. A ``?" next to a parameter signifies that the parameterized complexity of the problem for that parameter remains open.}
\renewcommand{\arraystretch}{2}
\setlength\tabcolsep{3.2pt}
\footnotesize
{
\centering
\begin{tabular}{|l|l|l|l|l|l|}\hline
    &  AppAdd & AppDel & $r$-VoteChg & $r$-VoteAddChg & $r$-VoteDelChg \\ \hline

{\av}&
{\poly} [{\thm}\ref{thm-appadd-av-poly}]&
{\poly} [{\thm}\ref{thm-appdel-av-poly}]&
{\nphshort} [$k=1 \wedge r\geq 4$, {\thm}\ref{thm-pvc-av-nph}]&
{\wahshort}$^{\ell, k}$ [$\abs{J}=1 \wedge r\geq 2$, {\thm}\ref{thm-vac-av-wah}]&
{\nphshort} [$k=1 \wedge r\geq 3$, {\thm}\ref{thm-vdc-av-nph-r-3-k-1}]\\

&
&
&
{\wahshort}$^{\ell, k}$ [$\abs{J}=1 \wedge r\geq 3$, {\thm}\ref{thm-vc-av-wa-hard-k-ell}] &
{\fpt}$^{m, n[r=m]}$ [{\thm}\ref{thm-fpt-m}, Cor.~\ref{thm-vc-av-vac-fpt-n-r-2m}] &
{\fpt}$^{|J|, m, n}$ [{\thms}\ref{thm-vdc-av-fpt-J}, \ref{thm-vdc-av-fpt-n}]\\

&
&
&
{\fpt}$^{m, n[r=m]}$ [{\thm}\ref{thm-fpt-m}, Cor.~\ref{thm-vc-av-vac-fpt-n-r-2m}]\hide{, {\poly} ($r$=1)}&
\hide{{\poly} ($r$=1)}{\poly} [$k=1$, {\thm}\ref{thm-vac-av-poly-k-1}]&
{\poly} [$r$=1, {\thm}\ref{thm-vdc-av-poly-r=1}], $\ell$ ?\\  \hline

{\sav}&
{\nphshort} [$k=1$, {\thm}\ref{thm-appadd-sav-nsav-np-hard}]&
{\nphshort} [$k=1$, {\thm}\ref{thm-appdel-sav-nsav-nph-k=1}]&
{\nphshort} [$k=1 \wedge r\geq 4$, {\thm}\ref{thm-vc-sav-nsav-nph-r-4}]&
{\nphshort} [$k=1\wedge r\geq 1$, {\thm}\ref{thm-vac-sav-nph-k-1-r-1}]&
{\nphshort} [$k=1\wedge r\geq 3$, {\thm}\ref{thm-vdc-sav-nsav-nph-r-3-k-1}]\\

&
{\fpt}$^m$ [{\thm}\ref{thm-fpt-m}]&
{\wahshort}$^{\ell, k}$ [$\abs{J}=1$, {\thm}\ref{thm-appdel-sav-nsav-wah-l-k}]&
{\wahshort}$^{\ell, k}$ [$\abs{J}=1 \wedge r\geq 1$, {\thm}\ref{thm-pvc-sav-nsav-wah-l-k-r-1}]&
{\fpt}$^m$ [{\thm}\ref{thm-fpt-m}]&
{\wahshort}$^{\ell, k}$ [$\abs{J} =1\wedge r\geq 1$, {\thm}\ref{thm-vdc-sav-nsav-wah-l-k-r-1}]\\

&$\ell$, $\abs{\discset}$ ?
&{\fpt}$^m$ [{\thm}\ref{thm-fpt-m}]
&{\fpt}$^m$\hide{, $n$~[$r=m$]} [{\thm}\ref{thm-fpt-m}]
&$\ell$, $\abs{\discset}$ ?
&{\fpt}$^m$ [{\thm}\ref{thm-fpt-m}]\\ \hline

{\nsav}&
{\nphshort} [$k=1$, {\thm}\ref{thm-appadd-sav-nsav-np-hard}]&
{\nphshort} [$k=1$, {\thm}\ref{thm-appdel-sav-nsav-nph-k=1}]&
{\nphshort} [$k=1 \wedge r\geq 4$, {\thm}\ref{thm-vc-sav-nsav-nph-r-4}]&
{\nphshort} [$k=1\wedge r\geq 1$, {\thm}\ref{thm-vac-sav-nph-k-1-r-1}]&
{\nphshort} [$k=1\wedge r\geq 3$, {\thm}\ref{thm-vdc-sav-nsav-nph-r-3-k-1}]\\

&
{\fpt}$^m$ [{\thm}\ref{thm-fpt-m}]&
{\wahshort}$^{\ell, k}$ [$\abs{J}=1$, {\thm}\ref{thm-appdel-sav-nsav-wah-l-k}]&
{\wahshort}$^{\ell, k}$ [$\abs{J}=1 \wedge r\geq 1$, {\thm}\ref{thm-pvc-sav-nsav-wah-l-k-r-1}]&
{\fpt}$^m$ [{\thm}\ref{thm-fpt-m}]&
{\wahshort}$^{\ell, k}$ [$\abs{J}=1 \wedge r\geq 1$, {\thm}\ref{thm-vdc-sav-nsav-wah-l-k-r-1}]\\

&$\ell$, $\abs{\discset}$ ?
&{\fpt}$^m$ [{\thm}\ref{thm-fpt-m}]
&{\fpt}$^m$ [{\thm}\ref{thm-fpt-m}]
&$\ell$, $\abs{\discset}$ ?
&{\fpt}$^m$ [{\thm}\ref{thm-fpt-m}]\\ \hline

{\vccav}&
\multicolumn{2}{c|}{{\wahshort}$^k$ [$|J|=1 \wedge \ell=0$, {Thm.~}\ref{thm-NWD-ccav--wahard}]}&
\multicolumn{3}{c|}{{\wahshort}$^k$ [$|J|=1 \wedge \ell=0\wedge r\geq 0$, {Thm.~}\ref{thm-NWD-ccav--wahard}]}\\

&
\multicolumn{2}{c|}{{\fpt}$^m$ [{\thm}\ref{thm-fpt-m}]}&
\multicolumn{3}{c|}{{\fpt}$^m$ [{\thm}\ref{thm-fpt-m}]} \\ \hline

{\pav}&
\multicolumn{2}{c|}{{\wahshort}$^k$ [$|J|=1 \wedge \ell=0$, {Thm.~}\ref{thm-NWD-ccav--wahard}]}&
\multicolumn{3}{c|}{{\wahshort}$^k$ [$|J|=1 \wedge \ell=0\wedge r\geq 0$, {Thm.~}\ref{thm-NWD-ccav--wahard}]}\\

&
\multicolumn{2}{c|}{{\fpt}$^m$ [{\thm}\ref{thm-fpt-m}]}&
\multicolumn{3}{c|}{{\fpt}$^m$ [{\thm}\ref{thm-fpt-m}]} \\ \hline
\end{tabular}
\label{tab-results-summary}
}
\end{sidewaystable}

\subsection{Some Hard Problems Establishing Our Results}
We assume the reader is acquainted with the fundamentals of computational complexity, parameterized complexity, and graph theory. For those readers who are not, we recommend referring to~\cite{DBLP:conf/lata/Downey12,DBLP:journals/interfaces/Tovey02,Douglas2000}. Our hardness results are based on reductions from the following problems.

\EP
{\prob{Restricted Exact Cover by Three Sets} (\prob{RX3C})}
{A universe~$\xs$ of cardinality~$3\xsize$ for some integer~$\xsize$, along with a collection~$\xc$ of subsets of~$\xs$, where each element of~$\xc$ has cardinality~$3$, and each element of~$\xs$ appears in exactly three elements of~$\xc$.}
{Does~$\xc$ contain an exact set cover of~$\xs$, i.e., a subcollection $\xc'\subseteq \xc$ of cardinality~$\xsize$ such that every element of~$\xs$ appears in exactly one element of~$\xc'$?}

It is known that the {\prob{RX3C}} problem is {\nph}~\cite{DBLP:journals/tcs/Gonzalez85}. 
Note that for every {\prob{RX3C}} instance $(\xs, \xc)$ where $\abs{\xs}=3\xsize$, it holds that $\abs{\xc}=\abs{\xs}$.

An independent set in a graph is a subset of pairwise nonadjacent vertices.
\EP{\prob{$\kappa$-Independent Set}}
{A graph~$G$ and a positive integer~$\kappa$.}
{Does~$G$ have an independent set of~$\kappa$ vertices?}

A clique in a graph is a subset of pairwise adjacent vertices.

\EP{\prob{$\kappa$-Clique}}
{A graph $G$ and a positive integer~$\kappa$.}
{Does $G$ have a clique of $\kappa$ vertices?}

It is well-known that both the  \prob{$\kappa$-Independent Set} problem and the {\prob{$\kappa$-Clique}} problem are {\wah} with respect to the solution size~$\kappa$~\cite{DBLP:conf/coco/DowneyF92}. Moreover, both problems remain {\wah}  when restricted to regular graphs~\cite{DBLP:journals/cj/Cai08}.

For a graph $G=(U, A)$ and a subset $U'\subseteq U$, we denote by $G[U']$ the subgraph of~$G$ induced by~$U'$. For an integer~$i$, let $[i]=\{j\in \mathbb{N} \setmid 1\leq j\leq i\}$ be the set of all positive integers at most~$i$.

\section{Bribery Under Intractable Rules}
\label{sec-ccav-pav}
We commence our exploration with {\vccav} and {\pav}. In contrast to other rules studied in this paper, the computation of winners for {\vccav} and {\pav} is {\nph}~\cite{DBLP:conf/atal/AzizGGMMW15,DBLP:journals/jair/BetzlerSU13}. Nevertheless, these rules remain appealing for several reasons. First, they satisfy several proportional properties that many other rules often fail to meet~\cite{DBLP:journals/scw/AzizBCEFW17,DBLP:conf/atal/FernandezF19,DBLP:conf/aaai/FernandezELGABS17}. Second, numerous {\fpt}-algorithms and competitive approximation algorithms have been reported for calculating winners under these rules~\cite{DBLP:conf/atal/AzizGGMMW15,DBLP:journals/algorithmica/GuptaJST23,DBLP:journals/jair/SkowronF17,DBLP:journals/ai/SkowronFS15,DBLP:journals/aamas/YangW23}. Additionally, polynomial-time algorithms for restricted domains have also been derived~\cite{DBLP:conf/aaai/Peters18,DBLP:conf/ijcai/Yang19a}. 

We show that all bribery problems defined in this paper are {\wah} under {\vccav} and {\pav} with respect to the size of the winning committees, even when there is only one distinguished candidate ($\abs{J}=1$), the budget is~$0$ ($\ell=0$), and every vote approves at most two candidates. We define this special case as {\sc{Non-Winner Determination}} for~$f$ ({\prob{NWD}}-$f$), formally defined below.

\EP{{\prob{NWD}}-$f$}
{An election~$(C, V)$, a distinguished candidate~$p\in C$, and an integer~$k$.}
{Is~$p$  not included in any winning $k$-committees of~$(C, V)$ under a rule~$f$?}  

Our {\wahns} results are based on reductions from the {\prob{$\kappa$-Independent Set}} problem restricted to regular graphs.

\begin{theorem}
\label{thm-NWD-ccav--wahard}
{\emph{\prob{NWD-{\vccav}}}} is {\emph\wah} with respect to the parameter~$k$, even when every vote approves at most two candidates.
\end{theorem}

\begin{proof}
Let $(G, \kappa)$ be an instance of the {\prob{$\kappa$-Independent Set}} problem, where $G=(U, A)$ is a~$\dd$ regular graph with $\dd>0$. We create an instance of {\prob{NWD-{\vccav}}}, denoted $(E, p, k)$, as follows. For each vertex~$u\in U$ in~$G$, we create one candidate~$c(u)$. In addition, we create a distinguished candidate~$p$. Let $C=\{c(u)\setmid u\in U\}\cup \{p\}$. 
The multiset~$V$ of votes is formed as follows. For each edge $\edge{u}{u'}\in A$, we create one vote $v(u, u')=\{c(u), c(u')\}$. Furthermore, we create ${\dd}-1$ votes, each exclusively approving the distinguished candidate~$p$. Finally, we set $k=\kappa$. Let $E=(C, V)$. The construction clearly runs in polynomial time. It remains to show the correctness.

$(\Rightarrow)$ Assume that~$G$ has an independent set of~$\kappa$ vertices. In this case, every $k$-committee corresponding to an independent set of~$\kappa$ vertices satisfies the maximum number of $\kappa\cdot {\dd}=k\cdot \dd$ votes. However, every $k$-committee containing the distinguished candidate~$p$ satisfies at most $(k-1)\cdot {\dd}+({\dd}-1)=k\cdot {\dd}-1$ votes. Therefore,~$p$ cannot be included in any winning $k$-committees.

$(\Leftarrow)$ Assume that~$G$ does not have any independent set of~$\kappa$ vertices. We show below that there exists at least one winning $k$-committee which contains the distinguished candidate~$p$. If~$p$ is included in all winning $k$-committees, we are done. Suppose there is a winning $k$-committee~$w$ which contains only candidates corresponding to a set of~$k$ vertices in~$G$. Since~$G$ lacks an independent set of~$\kappa$ vertices, there must be at least one edge~$\edge{u}{u'}$ in~$G$ such that both~$c(u)$ and~$c(u')$ are in~$w$.
Moreover, with exactly ${\dd}-1$ votes approving only the distinguished candidate~$p$, the committee $w'=(w\setminus \{c(u)\})\cup \{p\}$ satisfies at least $\scoreE{w}{E}{\vccav}-({\dd}-1)+({\dd}-1)=\scoreE{w}{E}{\vccav}$ votes. This implies that~$w'$ is also a winning $k$-committee.
\end{proof}

For {\pav}, we obtain the same result.

\begin{theorem}
\label{thm-NWD-pav--wahard}
{\emph{\prob{NWD-{\pav}}}} is {\emph\wah} with respect to the parameter~$k$, even when every vote approves at most two candidates.
\end{theorem}

\begin{proof}
The proof for {\pav} follows a similar logic to the proof of Theorem~\ref{thm-NWD-ccav--wahard}. Given an  instance  $(G, \kappa)$ of the {\prob{Independent Set}} problem, where $G=(U, A)$ is a~${\dd}$-regular graph for some positive integer~${\dd}$, we create an instance of {\prob{NWD-{\pav}}}, denoted $((C, V), p, k)$, as follows. First, we create the same candidates as in the proof of Theorem~\ref{thm-NWD-ccav--wahard}. Then, for each edge~$\edge{u}{u'}\in A$, we create two votes~$v_1(u, u')$ and~$v_2(u, u')$, each approving exactly~$c(u)$ and~$c(u')$. Additionally, we create $2d-1$ votes, each solely approving the distinguished candidate~$p$. Finally, we set $k=\kappa$. The correctness arguments parallel the proof of Theorem~\ref{thm-NWD-ccav--wahard}. If the graph~$G$ has an independent set of~$\kappa$ vertices, then any $k$-committee corresponding to an independent set of~$\kappa$ vertices has a {\pav} score of $2\kappa\cdot d=2k\cdot \dd$. However, every $k$-committee containing the distinguished candidate has a {\pav} score of at most $(2d-1)+2 d\cdot (k-1)=2k\cdot d-1$. Therefore,~$p$ cannot be included in any winning $k$-committee. Let us consider the other direction now. If~$G$ lacks any independent set of~$\kappa$ vertices, then, akin to the approach in the proof of Theorem~\ref{thm-NWD-ccav--wahard}, one can check that any winning $k$-committee~$w$ without the distinguished candidate can be modified by replacing one specific candidate in~$w$ with~$p$ to obtain another winning $k$-committee.
\end{proof}

The above theorems offer us the following corollary.

\begin{corollary}
\label{cor-ccav-pav-w-hardness}
For each $f\in \{\text{{\emph\vccav}}, \text{{\emph\pav}}\}$, the problems {\emph\pappadd{f}}, {\emph\pappdel{f}}, {\emph\pvc{f}}, {\emph\pvac{f}}, and {\emph\pvdc{f}} are {\emph\wah} when parameterized by~$k$. These assertions hold even when $\abs{\discset}=1$, the budget of the briber is $\ell=0$, and every vote approves at most two candidates. For~{\emph\pvc{f}},~{\emph\pvac{f}}, and~{\emph\pvdc{f}}, the {\emph\wahns} holds for all~$r\geq 0$.
\end{corollary}

\section{Bribery Under Polynomial-Time Rules}
In this section, we investigate destructive bribery for {\av}, {\sav}, and {\nsav}.
We begin by examining a relationship between {\sav} and {\nsav} elections, enabling us to automatically derive hardness results for {\nsav} based on those for {\sav}.
Assume that we have a hardness result for {\sav} obtained via a reduction that constructs an election with the following property: for any two candidates with different {\sav} scores, the absolute difference in their scores exceeds a certain threshold. To show  hardness for {\nsav}, we introduce a sizable set of dummy candidates, none of whom are approved by any voter (and none of whom are distinguished candidates). The abundance of the dummy candidates and the substantial minimum score gap between candidates ensure that {\nsav} scores remain dominated by their {\sav} scores. Specifically, a candidate has a higher (or lower) {\sav} score than another {\iff} the same holds for their {\nsav} scores after incorporating the dummy candidates.

This relationship is formalized in Lemma~\ref{lem-relation-sav-nsav} below. 

\begin{lemma}
\label{lem-relation-sav-nsav}
Let $E = (C, V)$ be an election of $m$ candidates and $n$ votes with $m \geq 2$.  
Let $c, c' \in C$ be two candidates such that the absolute difference in their {\sav} scores in $E$ is at least $\frac{1}{m \cdot (m-1)}$.  
Let $D$ be a set of at least $n\cdot m^2$ candidates disjoint from $C$, and let $E' = (C \cup D, V)$.  
Then,  
$\scoreE{c}{E}{\sav} > \scoreE{c'}{E}{\sav}$ {\iff} $\scoreE{c}{E'}{\nsav} > \scoreE{c'}{E'}{\nsav}$.
\end{lemma}

\begin{proof}
In election $E'$, the candidates in $D$ receive no approvals from any vote in $V$.  
Thus, the {\nsav} score of any candidate $c \in C$ in $E'$ is its {\sav} score in $E$ minus  
$\sum_{v \in V, c \notin v} \frac{1}{m + \abs{D} - \abs{v}}$. 
Since $\abs{D} \geq n\cdot m^2$ and $\abs{v} \leq m-1$, it follows that  
\[\sum_{v \in V, c \notin v} \frac{1}{m + \abs{D} - \abs{v}} \leq \frac{n}{m+n\cdot m^2-(m-1)}\leq \frac{1}{m^2+1/n} < \frac{1}{m\cdot (m-1)}.\]
The lemma follows.
\end{proof}

All hardness results for {\nsav} in this paper stem from modifications of the  reductions for the same problems under {\sav}  by adding dummy candidates as discussed above. The correctness of the modifications are ensured by Lemma~\ref{lem-relation-sav-nsav} and the following fact: in all our reductions for {\sav}, the constructed elections satisfy the property that for any two candidates with different {\sav} scores, the absolute difference in their scores is greater than $\frac{1}{m\cdot (m-1)}$.\footnote{One can check that in all reductions, it is optimal not to add dummy candidates in the sets of approved candidates in votes.}

In the following, we divide our discussions into several subsections, each of which is dedicated to a concrete bribery problem.

\subsection{Approval Addition}
In this section, we study the atomic operation AppAdd. We show that among the five rules, AV is the sole rule that admits a polynomial-time algorithm.

\begin{theorem}
\label{thm-appadd-av-poly}
{\emph\pappadd{\emph\av}} is polynomial-time solvable.
\end{theorem}

\begin{proof}
Let $I=(E, J, k, \ell)$ be an instance of  {\pappadd{\av}}, where $E=(C, V)$ and $J\subseteq C$. Let~$m=\abs{C}$ and~$n=\abs{V}$ denote the number of candidates and the number of votes in~$E$, respectively. Observe that if some candidate from $\discset$ is included in all votes,~$I$ is a {\noins}. Therefore, in the following, we consider the case where none of~$\discset$ is included in all votes. We derive an algorithm as follows. First, we compute the {\av} scores of all candidates, and let~$c^{\star}$ be a candidate from~$\discset$ such that $\scoreE{c^{\star}}{E}{\av} \geq \scoreE{c}{E}{\av}$ for all $c\in \discset$. Let 
\[C^{>}(c^{\star})=\{c\in C\setminus\discset \setmid \scoreE{c}{E}{\av} > \scoreE{c^{\star}}{E}{\av}\}\] 
be the set of all nondistinguished candidates whose {\av} scores are higher than that of~$c^{\star}$. 

If $\abs{C^{>}(c^{\star})} \geq k$, we conclude that~$I$ is a {\yesins}, since every distinguished candidate has a lower AV score than any candidate in $C^{>}(c^{\star})$. As $C^{>}(c^{\star})$ contains at least $k$ candidates, none of the distinguished candidates can be a winner.

Now, consider the case where $\abs{C^{>}(c^{\star})} < k$. Adding candidates to votes never decreases the {\av} scores of any candidate, and it is always optimal to avoid adding distinguished candidates to any vote. Thus, the problem reduces to determining whether we can perform at most~$\ell$ AppAdd operations to ensure that at least~$k$ candidates from $C \setminus J$ achieve an {\av} score of at least $\scoreE{c^{\star}}{E}{\av} + 1$. For each candidate $c\in C\setminus (\discset \cup C^{>}(c^{\star}))$, define
\[{\sf{diff}}(c)=\scoreE{c^{\star}}{E}{\av}+1-\scoreE{c}{E}{\av}\] 
as the minimum number of AppAdd operations  required to increase~$c$'s {\av} score to at least $\scoreE{c^{\star}}{E}{\av} + 1$. We order the candidates in $C \setminus (\discset \cup C^{>}(c^{\star}))$ in nondecreasing order of~${\sf{diff}}(c)$ and define~$A$ as the set containing the first $k - \abs{C^{>}(c^{\star})}$ candidates in this ordering.

If $\sum_{c \in A} {\sf{diff}}(c) \leq \ell$, we conclude that~$I$ is a {\yesins}; otherwise,~$I$ is a {\noins}. 
The correctness of this step follows from the following reasoning: 
\begin{itemize}
    \item In the first case, for each candidate $c \in A$, we arbitrarily select ${\textsf{diff}}(c)$ votes from 
    $\{v \in V \setmid c\not\in v\}$, add~$c$ to these votes, and decrement~$\ell$ by ${\textsf{diff}}(c)$ accordingly. After performing these operations, we have $\ell \geq 0$, all candidates in $A \cup C^{>}(c^{\star})$ have strictly higher scores than~$c^{\star}$, and $\abs{A \cup C^{>}(c^{\star})} = k$. 
    \item Conversely, if $\sum_{c \in A}{\sf{diff}}(c) > \ell$, then by the definition of $A$, more than~$\ell$ AppAdd operations would be needed to ensure that at least $k - \abs{C^{>}(c^{\star})}$ candidates from $C \setminus (\discset \cup C^{>}(c^{\star}))$ reach a score of at least $\scoreE{c^{\star}}{E}{\av} + 1$. 
\end{itemize}
It is clear that the above algorithm runs in polynomial time with respect to the size of $I$. This completes the proof of the theorem.
\end{proof}

An essential foundation for the correctness of the algorithm in the proof of Theorem~\ref{thm-appadd-av-poly} is the observation that adding candidates in a vote does not change the scores of other candidates. This characteristic enables a greedy solution to the instance. However, in {\sav} and {\nsav} voting, the addition of a candidate in a vote increases the score of that candidate while decreasing the scores of other candidates in the same vote. This divergence in behavior between {\av} and {\sav}/{\nsav} fundamentally distinguishes the complexity of the bribery problems under these rules.

\begin{theorem}
\label{thm-appadd-sav-nsav-np-hard}
{\emph\pappadd{{\emph\sav}}} and {\emph\pappadd{{\emph\nsav}}}  are {\emph\nph} even if $k=1$.
\end{theorem}

\begin{proof}
We provide the proof for {\sav} via a reduction from the {\prob{RX3C}} problem. The reduction for {\nsav} is a modification of the reduction for {\sav} based on Lemma~\ref{lem-relation-sav-nsav}.

Let $I=(\xs, \xc)$ be an instance of {\prob{RX3C}}, where $\abs{\xs}=\abs{\xc}=3\xsize$. We assume that $\xsize\geq 6$ and that~$\xsize$ is even, which does not affect the complexity of the problem.\footnote{If $\xsize< 6$, the problem can be solved in polynomial time. For odd~$\xsize$, we can obtain an equivalent instance by duplicating the given instance. In this transformation, the original instance has an exact set cover of cardinality~$\xsize$ if and only if the new instance has an exact set cover of cardinality~$2\xsize$.
}
We create an instance denoted by $g(I)=((C, V), \discset, \ell, k)$ of {\pappadd{{\sav}}} as follows.

For each $\xse\in \xs$, we create one candidate~$c(\xse)$. Let $C(\xs)=\{c(\xse) \setmid \xse\in \xs\}$. In addition, we create a candidate denoted by~$\p$. Let $C=C(\xs)\cup \{\p\}$, let $\discset=C(\xs)$, let $k=1$, and define $\ell=\xsize$. The multiset~$V$ comprises the following votes.
First, we create $\frac{3}{4} {\xsize}^2-3{\xsize}$ votes, each approving all candidates except~$\p$. As $\xsize\geq 6$ and~$\xsize$ is even,~$\frac{3}{4} {\xsize}^2-3{\xsize}$ is a positive integer. Additionally, for each $\xce\in \xc$, we create one vote~$v(\xce)=\{c(\xse) \setmid \xse\in \xce\}$ approving exactly the three candidates corresponding to~$\xce$. This completes the construction of~$g(I)$, which can be done in  polynomial time. In the election~$(C, V)$, the {\sav} score of~$\p$ is~$0$, and that of every other candidate is  \[\left(\frac{3}{4} \xsize^2-3{\xsize}\right)\cdot \frac{1}{3\xsize}+3\times \frac{1}{3}=\frac{\xsize}{4}.\] It remains to show the correctness of the reduction. Since $C = \discset \cup \{\p\}$ and $k = 1$, the problem $g(I)$ is equivalent to determining whether at most $\ell = \xsize$ AppAdd operations can be performed on $(C, V)$ to make $p$ the unique winner.

$(\Rightarrow)$ Assume that $\xc'\subseteq \xc$ is an exact $3$-set cover of~$\xs$, i.e., $\abs{\xc'}=\xsize$ and every $\xse\in \xs$ appears in exactly one element of~$\xc'$. Consider the election~$E$ obtained from~$(C, V)$ by adding~$\p$ to the vote~$v(\xce)$ for every $\xce\in \xc'$. Since $\abs{\xc'} = \xsize$, this process results in exactly $\xsize=\ell$ additions.  In~$E$, the votes approving~$\p$ are precisely those corresponding to~$\xc'$. Since each of these~$\xsize$ votes approves four candidates in~$E$, the {\sav} score of~$p$ in~$E$ is~$\frac{\xsize}{4}$. For each candidate $c(\xse)$ where $\xse \in \xs$, their {\sav} score decreases by $\frac{1}{3} - \frac{1}{4} = \frac{1}{12}$ when $\p$ is added to a vote $v(\xce)$ such that $\xse \in \xce$ for some $\xce \in \xc'$. Since~$\xc'$ is an exact set cover, there is exactly one such vote for each candidate~$c(\xse)$. Therefore,  the {\sav} score of~$c(\xse)$ where $\xse\in \xs$ becomes $\frac{\xsize}{4}-\frac{1}{12}$ in~$E$. As a result,~$p$ uniquely wins~$E$, meaning that $g(I)$ is a {\yesins}.

$(\Leftarrow)$ Assume that we can perform at most $\ell=\xsize$ AppAdd operations on~$(C, V)$ so that~$p$ uniquely wins the resulting election. Precisely, there exists $V'\subseteq V$, a multiset $V''$ of $\abs{V'}$ votes, and a one-to-one correspondence $h: V'\rightarrow V''$ such that the following conditions are satisfied:
\begin{enumerate}
    \item[(1)] For every $v\in V'$, it holds that $v\subsetneq h(v)$, i.e., $v$ is changed into~$h(v)$ by performing at least one AppAdd operation;
    \item[(2)] The total number of AppAdd operations performed satisfies $\sum_{v \in V'} \abs{h(v) \setminus v} \leq \xsize$. 
    \item[(3)] $p$ uniquely wins the resulting election $E=(C, V\setminus V'\cup V'')$. 
\end{enumerate}
Let $V(\xc)=\{c(\xce) \setmid \xce\in \xc\}$, and let~$t=\abs{V'\cap (V\setminus V(\xc))}$ be the number of votes, among the $\frac{3}{4} {\xsize}^2-3{\xsize}$ votes approving~$\discset$, that are changed. 
We prove that all votes in $V'$ originate from $V(\xc)$, as reflected by the following claim.  

\begin{claim}  
\label{claim-appadd-sav-nsav-np-hard}  
    $t = 0$.  
\end{claim}  

\begin{proof}
Assume, for the sake of contradiction, that $t>0$. Then, by Condition~(1) and Condition~(2) above, at most $\xsize-1$ votes from~$V(\xc)$ are changed, i.e., $\abs{V'\cap V(\xc)}\leq \xsize-1$. This implies that there exists at least one distinguished candidate~$c(\xse)$, where $\xse\in \xs$, such that none of the three votes~$v(\xce)$ with $\xse\in \xce\in \xc$ is contained in~$V'$. Consequently, the {\sav} score of~$c(\xse)$ in~$E$ is
\[\left(\frac{3}{4} {\xsize}^2-3{\xsize-t}\right)\cdot \frac{1}{3\xsize}+\frac{t}{3\xsize+1}+1=\frac{\xsize}{4}-\frac{t}{3\xsize}+\frac{t}{3\xsize+1}.\]
However, by Condition~(1), the {\sav} score of~$\p$ in~$E$ can be at most
\[\frac{t}{3\xsize+1}+\frac{\xsize-t}{4},\] which is strictly smaller than that of~$c(\xse)$. This contradicts Condition~(3). 
We can therefore conclude that $t=0$.
\end{proof}

By Claim~\ref{claim-appadd-sav-nsav-np-hard}, we know that $V'\subseteq V(\xc)$. Let $\xc'=\{\xce\in \xc \setmid v(\xce)\in V'\}$. Observe that when some~$v(\xce)$, where $\xce\in \xc$, is to be modified, it is optimal to add only~$p$ to the vote in order to achieve the goal of making $p$ the unique winner.
Therefore, we may assume that for every $v(\xce)\in V'$, it holds that $h(v(\xce))=v(\xce)\cup \{p\}$. Moreover, it is optimal to fully use the budget $\ell=k$. Thus, we assume that $\abs{V'}=\ell=\kappa$. It follows that $\abs{\xc'}=\kappa$, and the {\sav} score of~$p$ in~$E$ is~$\frac{\xsize}{4}$. Since $p$ uniquely wins~$E$, and every candidate~$c(\xse)$, where $\xse\in \xs$, has an original {\sav} score of $\frac{\xsize}{4}$, the {\sav} score of $c(\xse)$ must be decreased in the final election~$E$. By the construction above, this means that for every candidate~$c(\xse)$, where $\xse\in \xs$, at least one vote~$v(\xce)$ with $\xse\in \xce\in \xc$ is contained in~$V'$. It follows that~$\xc'$ is an exact $3$-set cover of~$\xs$. 

To establish the hardness of {\pappadd{{\nsav}}}, we augment the above constructed election $(C, V)$ by introducing an additional set of $\abs{V}\cdot \abs{C}^2$ candidates, which are set as distinguished candidates (i.e., placed into the set~$J$); other components of the reduction remain unchanged. The correctness for the $\Rightarrow$ direction is similar.  For the $\Leftarrow$ direction, note that in the original election $(C, V)$, there are only two distinct {\sav} scores: $0$ and $\frac{\xsize}{4}$. It is therefore evident that the precondition in Lemma~\ref{lem-relation-sav-nsav} holds. The correctness of the $\Leftarrow$ direction is then ensured by Lemma~\ref{lem-relation-sav-nsav} and the observation that if the constructed instance of {{\pappadd{{\nsav}}}} is a {\yesins}, it is optimal to change $\ell$ votes from $V(\xc)$ as described in the proof of the $\Leftarrow$ direction above.
\end{proof}

\subsection{Approval Deletion}
This section examines the atomic operation AppDel. 
For a sequence $\textsf{OP}$ of AppDel operations and an election~$E=(C, V)$, we denote by $\textsf{OP}(E)$ the election obtained by applying all operations in~$\textsf{OP}$ to~$E$. 
Additionally, we slightly abuse notation by letting $\abs{\textsf{OP}}$ represent the number of AppDel operations in~$\textsf{OP}$. Furthermore, for a vote $v\in V$, we use~$\textsf{OP}(v)$ to denote the vote resulting from applying the operations in~$\textsf{OP}$ to~$E$. 

For {\av},  we once again establish a polynomial-time solvability result.

\begin{theorem}
\label{thm-appdel-av-poly}
{\emph{\pappdel{\av}}} is polynomial-time solvable.
\end{theorem}

\begin{proof}
Let $I = (E, \discset, k, \ell)$ be an instance of {\pappdel{\av}}, where $E = (C, V)$ is an election and $\discset \subseteq C$ is nonempty. Our algorithm proceeds as follows.

We compute the subset $A \subseteq C \setminus J$ of nondistinguished candidates approved by at least one vote from~$V$. If $\abs{A} \leq k - 1$, we immediately conclude that $I$ is a {\noins}. The reasoning is that, in this case, any $k$-committee of $E$ that contains $A$ and any other $k - \abs{A}$ candidates from $C \setminus A$ forms an AV winning $k$-committee of $E$. Given that $J\subseteq C\setminus A$ and $J \neq \emptyset$, we can conclude that there exists an AV winning $k$-committee of $E$ that intersects with $J$. Thus, the given instance is a {\noins}.

Next, our algorithm proceeds by exhaustively applying the following reduction rule. 

\begin{reductionrule} 
\label{reduction-rule-a}
Let~$c^{\star}\in J$ be a distinguished candidate with the maximum {\av} score among all distinguished candidates, i.e., $\scoreE{c^{\star}}{E}{\av} \geq \scoreE{c}{E}{\av}$ for all $c\in J$. If ${\scoreE{c^{\star}}{E}{\av}}\geq 1$ and there are at most $k-1$ candidates from~$C\setminus \discset$ with {\av} scores at least $\scoreE{c^{\star}}{E}{\av}+1$, remove~$c^{\star}$ from any arbitrary vote approving~$c^{\star}$, and decrease~$\ell$ by one.
\end{reductionrule}

We prove now that Reduction Rule~\ref{reduction-rule-a} is correct, i.e., applying this rule to any instance of {\emph{\pappdel{\av}}} results in an equivalent instance. 

\begin{claim}
\label{claim-reduction-rule}
    Reduction Rule~\ref{reduction-rule-a} is correct.
\end{claim}

\begin{proof}
    Let $c^{\star}$ be a distinguished candidate, as defined in Reduction Rule~\ref{reduction-rule-a}.  
    Assume that at most $k-1$ candidates from $C \setminus \discset$ have {\av} scores of at least $\scoreE{c^{\star}}{E}{\av} + 1$, and that ${\scoreE{c^{\star}}{E}{\av}} \geq 1$. 
   Let $I' = ((C, V'), \discset, k, \ell - 1)$ be the instance obtained from $I$ by applying Reduction Rule~\ref{reduction-rule-a} to~$c^{\star}$. Thus,~$V'$ is derived from~$V$ by removing~$c^{\star}$ from some vote, say $v^{\star} \in V$, where~$c^{\star}$ is approved, i.e., $V' = V \setminus \{v^{\star}\} \cup \{v^{\star} \setminus \{c^{\star}\}\}$. 
   Since ${\scoreE{c^{\star}}{E}{\av}} \geq 1$, such a vote~$v^{\star}$ must exist. Define $v^{\star-}=v^{\star}\setminus \{c^{\star}\}$, and let $E'=(C, V')$.

It is clear that if~$I'$ is a {\yesins}, then~$I$ is also a {\yesins}. 
Now consider the opposite direction. For clarity, in the following discussion, for a vote~$v$ and some candidate $c\in v$, we use $(v, c)$ to denote the AppDel operation where~$c$ is deleted from~$v$. 
Assume that~$I$ is a {\yesins}. Our goal is to ensure that none of the distinguished candidates are included in any winning $k$-committee by performing a limited number of AppDel operations. Therefore, it is not meaningful to perform any AppDel operations on nondistinguished candidates. We may thus assume that there exists a sequence~$\textsf{OP}$ of at most~$\ell$ AppDel operations that transforms~$E$ into an election $\widetilde{E}$ such that (1) each AppDel operation is performed on some vote $v\in V$ by deleting a distinguished candidate in $\discset$, and (2) none of the candidates in $J$ are included in any AV winning $k$-committees of~$\widetilde{E}$.   
By Condition~(1), we know that in $\widetilde{E}$, the number of nondistinguished candidates whose AV score exceeds~$\scoreE{c^{\star}}{E}{\av}$ is limited to $k - 1$. Then, Condition~(2) implies that $\textsf{OP}$ contains at least one AppDel operation in which $c^{\star}$ is removed from some vote in $V$. 
Our proof proceeds by considering two cases. 
\begin{itemize}
    \item Case~1. The AppDel operation $(v^{\star}, c^{\star})$ is in $\textsf{OP}$. 

    In this case, we define $\textsf{OP}'$ as the sequence of AppDel operations obtained from~$\textsf{OP}$ by removing the AppDel operation $(v^{\star}, c^{\star})$. Obviously,~$\textsf{OP}'$ contains at most $\ell - 1$ AppDel operations. Furthermore, after applying the operations in~$\textsf{OP}'$ to~$E'$, we obtain exactly the election $\widetilde{E}$. Then, by Condition~(2) above, we know that $I'$ is a {\yesins}.

    \item Case~2. The AppDel operation $(v^{\star}, c^{\star})$ is not in $\textsf{OP}$. 
    
    As discussed above,~$\textsf{OP}$ contains an AppDel operation $(v^{\circ}, c^{\star})$ for some $v^{\circ} \in V$. We know that $v^{\circ}\neq v^{\star}$. We define $\textsf{OP}'$ as the sequence of AppDel operations obtained from $\textsf{OP}$ by removing the AppDel operation $(v^{\circ}, c^{\star})$. Thus,~$\textsf{OP}'$ contains at most $\ell-1$ AppDel operations. Let $\widehat{E}=\textsf{OP}'(E')$. Moreover, for each vote $v\in V\setminus \{v^{\star}\}$, let $\widetilde{v}=\textsf{OP}(v)$, and let $\widehat{v}=\textsf{OP}'(v)$. By the relationship between~$\textsf{OP}'$ and~$\textsf{OP}$, and the relationship between~$E$ and~$E'$, the following hold:
    \begin{itemize}
        \item for every $v\in V\setminus \{v^{\circ}, v^{\star}\}$, we have $\widetilde{v}=\widehat{v}$, and 
        \item $\widetilde{v^{\circ}}=\widehat{v^{\circ}}\cup \{c^{\star}\}$.
    \end{itemize}
    Recall that $v^{\star-}=v^{\star}\setminus \{c^{\star}\}$ and $V' = V \setminus \{v^{\star}\} \cup \{v^{\star-}\}$. As a result, the AV scores of all candidates remain unchanged  in~$\widetilde{E}$ and~$\widehat{E}$. Therefore, if none of the distinguished candidates is included in any winning $k$-committee of~$\widetilde{E}$ (which holds by Condition~(2)), the same holds for~$\widehat{E}$. Consequently, $I$ is a {\yesins}.
\end{itemize}
    The proof for the claim is completed. 
\end{proof}

After exhaustively applying this rule, we conclude that~$I$ is a {\yesins} if and only if $\ell \geq 0$. The correctness of this step follows from Claim~\ref{claim-reduction-rule}.  

Since each application of Reduction Rule~\ref{reduction-rule-a} runs in polynomial time, and it can be applied at most $\abs{C} \cdot \abs{V}$ times, the overall algorithm runs in polynomial time.
\end{proof}

However, for {\sav} and {\nsav}, we encounter hardness results.

\begin{theorem}
\label{thm-appdel-sav-nsav-nph-k=1}
{\emph{\pappdel{{\sav}}}} and {\emph{\pappdel{{\nsav}}}} are {\emph\nph} even when $k=1$.
\end{theorem}

\begin{proof}
We present the proof for {\sav}, and the proof for {\nsav} can be obtained through a modification of the proof for {\sav} based on Lemma~\ref{lem-relation-sav-nsav}.

The proof is via a reduction from the {\prob{RX3C}} problem. Let $I=(\xs, \xc)$ be an instance of {\prob{RX3C}} with  $\abs{\xs}=\abs{\xc}=3\xsize>0$. Similar to the proof of Theorem~\ref{thm-appadd-sav-nsav-np-hard}, we assume that~$\xsize$ is even.
We create an instance $g(I)=((C, V), J, \ell, k)$, where $\ell=3\xsize$ and $k=1$, of {\pappdel{{\sav}}} as follows.
The candidate set is $C=\xs\cup \{p\}$. Let $\discset=\xs$. Regarding the votes, we first create a multiset~$V'$ of $\frac{9}{2} \xsize^2-2\xsize$ votes, each approving exactly the distinguished candidates in~$\discset$. As~$\xsize$ is a positive even integer, $\frac{9}{2} \xsize^2-2\xsize$ is a positive integer. Then, for each $\xce\in \xc$, we create one vote~$v(H)=\{p\}\cup H$ approving~$\p$ and the three candidates corresponding to the three elements of~$\xce$. Let $V(\xc)=\{v(\xce) \setmid \xce\in \xc\}$, let $V=V'\cup V(\xc)$, and define $E=(C, V)$.

The construction of~$g(I)$ can be completed in polynomial time. We prove the correctness of the reduction as follows. In the election~$E$, the {\sav} score of~$\p$ is $\frac{3\xsize}{4}$, and that of every other candidate is
\[\left(\frac{9\xsize^2}{2}-2\xsize\right)\cdot\frac{1}{3\xsize}+\frac{3}{4}=\frac{3 \xsize}{2}+\frac{1}{12}.\] 

$(\Rightarrow)$ Assume that there is an exact $3$-set cover $\xc'\subseteq \xc$ of~$\xs$. Consider the election~$E'$ obtained from~$E$ by the following modifications: for every $\xce\in\xc'$, we remove from~$v(H)$ the three candidates corresponding to the three elements in~$\xce$. As $\abs{\xc'}=\xsize$, we performed in total~$3\xsize$ AppDel operations. The modifications increase the {\sav} score of~$\p$ by~$\frac{3\abs{\xc'}}{4}=\frac{3\xsize}{4}$, and decrease the {\sav} score of each $\xse\in \xs$ by~$\frac{1}{4}$. So, in the resulting election~$E'$, the {\sav} score of~$p$ increases to $\frac{3\xsize}{2}$, and that of every $\xse\in \xs$ decreases to $\frac{3\xsize}{2}-\frac{1}{6}$. This implies that~$\p$ is the unique winner of~$E'$.

$(\Leftarrow)$ Assume that~$g(I)$ is an {\yesins}. Precisely, we can transform $E=(C, V)$ into a new election~$E'$ by performing a sequence of at most $\ell = 3\xsize$ AppDel operations such that none of the candidates from $\discset = \xs$ emerges as the winner of~$E'$, or equivalently, that~$p$ becomes the unique winner of~$E'$. We prove below that this implies that the given instance~$I$ of {\prob{RX3C}} is a {\yesins}. To this end, we study a series of claims that formally confirm the following intuitions: 
\begin{itemize}
    \item It is optimal to delete only distinguished candidates from the votes (Claim~\ref{claim-appdel-sav-nsav-nph-a}).
     
    \item It is optimal to fully utilize the deletion budget $\ell=3\xsize$ (Claim~\ref{claim-appdel-sav-nsav-nph-b}). 

    \item It is optimal to modify only the votes in~$V(\xc)$ and, moreover, when a vote~$v(\xce)\in V(\xc)$ is designated for modification, it is optimal to delete the three distinguished candidates corresponding to the elements in~$\xce$ (Claims~\ref{claim-xx} and~\ref{claim-last}). In general, the rationale is that deleting candidates from a vote in~$V'$ merely redistributes the {\sav} scores among the distinguished candidates, whereas deleting distinguished candidates from a vote in~$V(\xc)$ transfers their {\sav} scores to the sole nondistinguished candidate~$p$. Our meticulous design of the election ensures that if the increase in the {\sav} score of~$p$ is not evenly distributed across the decreases in the {\sav} scores of the distinguished candidates, then at least one distinguished candidate will ultimately have a higher {\sav} score than~$p$ in the final election. 
\end{itemize}

We now present the claims, introducing the necessary notations. 

For a vote~$v$ and a candidate $c\in v$, we use $(v, c)$ to denote the AppDel operation of deleting~$c$ from~$v$. 
 We say~$v$ (respectively,~$c$) is involved in the AppDel operation $(v, c)$. Furthermore,~$v$ (respectively,~$c$) is said to be involved in a sequence $\textsf{OP}$ of AppDel operations if it is involved in at least one AppDel operation in~$\textsf{OP}$. 

\begin{claim}
\label{claim-appdel-sav-nsav-nph-a}
   The instance $g(I)$ has a feasible solution where~$p$ is not involved.
\end{claim}

\begin{proof}
Let~$\textsf{OP}$ be any arbitrary feasible solution of~$g(I)$. If~$p$ is not involved in~$\textsf{OP}$, we are done. Otherwise, consider the sequence~$\textsf{OP}'$ obtained from~$\textsf{OP}$ by removing all AppDel operations in~$\textsf{OP}$ that involve~$p$. Clearly, $\abs{\textsf{OP}'}\leq \ell-1$. Moreover, the {\sav} score of~$p$ in $\textsf{OP}'(E)$ is strictly larger than that in $\textsf{OP}(E)$, while the {\sav} score of each other candidate is either the same or strictly smaller than that in~$\text{OP}(E)$. This implies that~$\textsf{OP}'$ is also a feasible solution of~$g(I)$. As~$p$ is not involved in~$\textsf{OP}'$, the claim is proved.  
\end{proof}

\begin{claim}
\label{claim-appdel-sav-nsav-nph-b}
   The instance $g(I)$ admits a feasible that consists of exactly $3\xsize$ AppDel operations.
\end{claim} 

\begin{proof}
    Let $\textsf{OP}$ be a feasible solution for~$g(I)$.  
    If $\abs{\textsf{OP}} = 3\xsize$, the claim holds. Otherwise, since~$\textsf{OP}$ is a feasible solution, we have $\abs{\textsf{OP}} \leq 3\xsize - 1$.  

    Notice that each vote in~$V(\xc)$ approves three distinguished candidates, and we have $\abs{\xc} = 3\xsize$. In other words, votes in~$V(\xc)$ collectively provide $9\xsize\geq 3\xsize$ approvals to the distinguished candidates. Therefore, after applying the operations in~$\textsf{OP}$ to~$E$, there exists at least one vote~$v(\xce)$,  $\xce \in \xc$, that still approves at least one distinguished candidate, say~$\xse$. Let~$X$ be the set of candidates still approved in~$v(\xce)$ after applying the operations in~$\textsf{OP}$ to~$E$. By Claim~\ref{claim-appdel-sav-nsav-nph-a}, we may assume that $p \in X$.  

    Now, consider the sequence $\textsf{OP}'$ obtained from $\textsf{OP}$ by adding the AppDel operation $(v(\xce), \xse)$. Clearly, $\abs{\textsf{OP}'} = \abs{\textsf{OP}} + 1 \leq 3\xsize$. Moreover, compared to~$\textsf{OP}(E)$, in~$\textsf{OP}'(E)$:
    \begin{itemize}
        \item The {\sav} scores of all candidates in $X \setminus \{\xse\}$, including~$p$, strictly increase by the same amount.
        \item The {\sav} score of $\xse$ strictly decreases.
        \item The {\sav} scores of all other candidates remain unchanged.
    \end{itemize} 

    Recall that~$p$ uniquely wins~$\textsf{OP}(E)$ (since $\textsf{OP}$ is a feasible solution of~$g(I)$). It follows that~$p$ also uniquely wins $\textsf{OP}'(E)$, ensuring that $\textsf{OP}'$ remains a feasible solution of~$g(I)$.   

    By iteratively applying this transformation to~$\textsf{OP}$ whenever $\abs{\textsf{OP}}<3\xsize$, we eventually obtain a feasible solution for~$g(I)$ consisting of exactly~$3\xsize$ AppDel operations.  
\end{proof}

We say that a vote is of type~$i$ (or a type-$i$ vote) with respect to a sequence~$\textsf{OP}$ of AppDel operations, if it is involved in exactly~$i$ AppDel operations in $\textsf{OP}$, meaning exactly~$i$ candidates are deleted from the vote according to $\textsf{OP}$.  

\begin{claim}
\label{claim-xx}
    Let $\mathcal{OP}$ be the set of all sequences of~$3\xsize$ AppDel operations. Then, the following statements hold: 
    \begin{enumerate}
        \item[(1)]  The maximum possible {\sav} score of $p$ in~$E$ after applying a sequence of~$3\xsize$ AppDel operations is~$\frac{3\kappa}{2}$. That is, \[\max_{\textsf{OP}\in \mathcal{OP}}\scoreE{p}{\textsf{OP}(E)}{\sav}\leq \frac{3\kappa}{2}.\]
        
        \item[(2)] If the {\sav} score of $p$ in some ${\textsf{OP}}(E)$, where $\textsf{OP}\in \mathcal{OP}$, attains~$\frac{3\kappa}{2}$, then, with respect to $\textsf{OP}$, there are no type-$1$ and type-$2$ votes in $V(\xc)$, and all votes in $V'$ are of type-$0$.
        
        \item[(3)] If the {\sav} score of $p$ in some $\textsf{OP}(E)$, where $\textsf{OP}\in \mathcal{OP}$, is strictly less than $\frac{3\kappa}{2}$, then \[\scoreE{p}{\textsf{OP}(E)}{\sav} \leq \frac{3\kappa}{2}-\frac{1}{6}.\] 
    \end{enumerate}
\end{claim}

\begin{proof}
We begin by proving the first statement.  
Let $\textsf{OP} \in \mathcal{OP}$ be a sequence of~$3\xsize$ AppDel operations that results in the highest {\sav} score of~$p$,  
meaning that $\scoreE{p}{\textsf{OP}(E)}{\sav} \geq \scoreE{p}{\textsf{OP}'(E)}{\sav}$ for all $\textsf{OP}' \in \mathcal{OP}$.

Note that deleting any candidate from a vote in~$V'$ does not affect the {\sav} score of~$p$, since none of the votes in~$V'$ approves~$p$. In contrast, deleting any distinguished candidate from a vote in~$V(\xc)$ strictly increases the {\sav} score of~$p$. Therefore, we may assume that no votes from~$V'$ are involved in~$\textsf{OP}$. Consequently, all votes from~$V'$ are of type-$0$. 
    
    Furthermore, removing~$p$ from any vote would only decrease its {\sav} score. Therefore, we may also assume that~$p$ is not involved in~$\textsf{OP}$. 
    
    For $i\in \{0, 1, 2, 3\}$, let~$t_i$ denote the number of type-$i$ votes in $V(\xc)$ with respect to~$\textsf{OP}$. Then, we have 
    \begin{equation}
    \label{eq-sav-p}
        \scoreE{p}{\textsf{OP}(E)}{\sav}=\frac{t_0}{4}+\frac{t_1}{3}+\frac{t_2}{2}+t_3.
    \end{equation}
    Moreover, we have 
    \begin{equation}
    \label{eq-sav-p-a}
        t_1+2t_2+3t_3=\abs{\textsf{OP}}=3\xsize.
    \end{equation} 
    Since each vote in~$V(\xc)$ approves exactly three distinguished candidates, and $\abs{V(\xc)}=3\xsize$, we obtain 
    \begin{equation}
        \label{eq-t0}
        t_0\leq 2\xsize.
    \end{equation}
    It follows that
    
        \begin{align}
         \scoreE{p}{\textsf{OP}(E)}{\sav} & =\frac{1}{12}\cdot \left(3{t_0}+4({t_1}+2{t_2}+3t_3)-2t_2\right) & \text{By}~\eqref{eq-sav-p}\notag \\
                              & =\frac{3t_0+12\xsize-2t_2}{12}                                   & \text{By}~\eqref{eq-sav-p-a} \label{eq-sav-b} \\
                              & \leq \frac{3\cdot 2\xsize+12\xsize-0}{12} = \frac{3\xsize}{2}.   & \text{By}~\eqref{eq-t0} \notag
    \end{align}

    This completes the proof of the first statement. 

Now, we prove the second statement. We have  shown above that if $p$ attains the highest possible {\sav} score with respect to some $\textsf{OP}$, then all votes in~$V'$ are of type-$0$. It remains to analyze the votes from $V(\xc)$. From~\eqref{eq-sav-b}, we know that $\scoreE{p}{\textsf{OP}(E)}{\sav}$ reaches its maximum possible value of $\frac{3\xsize}{2}$ if and only if $t_0=2\xsize$ and $t_2=0$. By the definition of $t_i$, $i\in \{0, 1, 2, 3\}$, we also have 
\begin{equation}\label{eq-sav-e}
    t_0+t_1+t_2+t_3=\abs{V(\xc)}=\abs{\xc}=3\xsize. 
\end{equation}
Substituting $t_0=2\xsize$ and $t_2=0$ into~\eqref{eq-sav-p-a} and~\eqref{eq-sav-e}, we obtain $t_1=0$ and $t_3=\xsize$. 
This confirms that there are no type-$1$ or type-$2$ votes in~$V(\xc)$ with respect to~$\textsf{OP}\in \mathcal{OP}$ whenever $\scoreE{p}{\textsf{OP}(E)}{\sav}=\frac{3\xsize}{2}$.

To prove the third statement, observe that~\eqref{eq-sav-b} achieves its second-largest possible value when $t_2= 1$ or $t_0=2\xsize-1$. A straightforward calculation shows that this value is bounded from above by $\frac{3\xsize}{2}-\frac{1}{6}$.
    \end{proof}

    \begin{claim}
    \label{claim-last}
        The instance $g(I)$ has a feasible solution $\textsf{OP}$ such that $\scoreE{p}{\textsf{OP}(E)}{\sav}=\frac{3\xsize}{2}$. 
    \end{claim}

    \begin{proof}
    By Claims~\ref{claim-appdel-sav-nsav-nph-a} and~\ref{claim-appdel-sav-nsav-nph-b}, we assume that $p$ is not involved in $\textsf{OP}$ and that $\abs{\textsf{OP}}=3\xsize$. 
    
If $\scoreE{p}{\textsf{OP}(E)}{\sav}=\frac{3\xsize}{2}$, the claim is immediately verified. 
Otherwise, by Statement~3 of Claim~\ref{claim-xx}, we have $\scoreE{p}{\textsf{OP}(E)}{\sav}\leq \frac{3\xsize}{2}-\frac{1}{6}$. We now show by contradiction that this cannot occur.

        Since~$\textsf{OP}$ is a feasible solution of~$g(I)$, it follows that~$p$ uniquely wins~$\textsf{OP}(E)$. This implies that the {\sav} score of every distinguished candidate in~$\textsf{OP}(E)$ is strictly less than $\frac{3\xsize}{2}-\frac{1}{6}$, i.e., $\scoreE{\xse}{\textsf{OP}(E)}{\sav}< \frac{3\xsize}{2}-\frac{1}{6}$ for all $\xse\in J$. 
        
        Now, we have 
        \begin{equation}\label{eq-p}
            \scoreE{p}{\textsf{OP}(E)}{\sav}-\scoreE{p}{E}{\sav}\leq  \left(\frac{3\xsize}{2}-\frac{1}{6}\right)-\frac{3\xsize}{4}= \frac{3\xsize}{4}-\frac{1}{6},
        \end{equation}
        and
        \begin{equation}\label{eq-dis}
            \sum_{\xse\in J}\left(\scoreE{\xse}{E}{\sav}-\scoreE{\xse}{\textsf{OP}(E)}{\sav}\right) 
            > \sum_{\xse\in J}\left(\left(\frac{3\xsize}{2}+\frac{1}{12}\right)-\left(\frac{3\xsize}{2}-\frac{1}{6}\right)\right)= \sum_{\xse\in J}\frac{1}{4}= \frac{3\xsize}{4}.
        \end{equation}

     Observe that~$\textsf{OP}$ cannot consists of solely the~$3\xsize$ operations of the form $(v, \xse)$ for some vote $v\in V'$, i.e., $\textsf{OP}(E)$ cannot be obtained from~$E$ by emptying a single vote~$v$ from~$V'$. The reason is that, in this case, the {\sav} score of~$p$ in~$\textsf{OP}(E)$ remains $\frac{3\xsize}{4}$, the same as in~$E$, while the {\sav} score of every other candidate decreases to $\frac{3\xsize}{2}+\frac{1}{12}-\frac{1}{3\xsize}>\frac{3\xsize}{4}$ in~$\textsf{OP}(E)$. This contradicts the assumption that~$p$ uniquely wins~$\textsf{OP}(E)$. 
     
     Furthermore, note that deleting any candidate from a vote in~$V'$ does not increase the total {\sav} score of the distinguished candidates. However, deleting a distinguished candidate from a vote in~$v(\xce)$ effectively redistributes the {\sav} score  that the removed candidate would have received from that vote among the remaining candidates in the vote.  

     It follows that the increase in the {\sav} score of~$p$ is exactly equal to the decrease in the total {\sav} score of all the distinguished candidates. That is, the following equality must hold: 
     \[\scoreE{p}{\textsf{OP}(E)}{\sav}-\scoreE{p}{E}{\sav}=\sum_{\xse\in J}\left(\scoreE{\xse}{E}{\sav}-\scoreE{\xse}{\textsf{OP}(E)}{\sav}\right).\]
        However, this contradicts~\eqref{eq-p} and~\eqref{eq-dis} above, completing the proof. 
    \end{proof}

    
By Claim~\ref{claim-last}, we know that~$g(I)$ admits a feasible solution~$\textsf{OP}$ such that the {\sav} score of~$p$ in~$\textsf{OP}(E)$ is exactly $\frac{3}{2}\xsize$. Consequently, for every $\xse\in\xs$, there must exist at least one vote $v(\xce)$, where $\xse\in \xce$ for some $\xce\in \xc$, that originally approves~$\xse$ but ceases to do so in $\textsf{OP}(E)$. In other words, for every $\xse\in \xs$, there exists at least one $\xce\in \xc$ such that $\xse\in \xce$ and the corresponding vote~$v(\xce)$ is involved in $\textsf{OP}$. This implies that the subcollection 
\[\xc'=\{\xce\in \xc \setmid v(\xce)~\text{is involved in}~\textsf{OP}\}\] 
of~$\xc$ forms a set cover of~$\xs$. By Statement~(2) of Claim~\ref{claim-xx}, we know that votes in $V(\xc)$ are either of type-$0$ or type-$3$ with respect to $\textsf{OP}$. Clearly, the number of votes in~$V(\xc)$ involved in $\textsf{OP}$ equals the number of type-$3$ votes, which is bounded above by $\frac{\abs{\textsf{OP}}}{3}=\xsize$. Thus, we conclude that $\abs{\xc'}=\xsize$, completing the proof that the instance~$I$ of {\prob{RX3C}} is a {\yesins}.
\end{proof}

We also establish {\wahns} results for {\sav} and {\nsav} for the parameters~$\ell$ and~$k$. This holds true even when there is only one distinguished candidate. Recall that Theorem~\ref{thm-appdel-sav-nsav-nph-k=1} is based on a reduction where the objective is to select only one winner, while the number of distinguished candidates equals to the number of candidates minus one.

\begin{theorem}
\label{thm-appdel-sav-nsav-wah-l-k}
{\emph{\pappdel{{\sav}}}}  and {\emph{\pappdel{{\nsav}}}} are {\emph{\wah}} with respect to both the parameter~$\ell$ and the parameter~$k$. Moreover, the results hold even when $\abs{J}=1$.
\end{theorem}

\begin{proof}
We prove the theorem by reducing from the {\prob{$\kappa$-Clique}} problem restricted to regular graphs.
We present the proof only for {\sav}. Later, it will become evident that the absolute score difference between any two candidates with distinct {\sav} scores in the structured election below satisfies the precondition of Lemma~\ref{lem-relation-sav-nsav}. 
Consequently, by introducing dummy candidates as suggested in Lemma~\ref{lem-relation-sav-nsav}, we can adapt the reduction for SAV to establish the {\wahns} of {{\pappdel{{\nsav}}}} with respect to $\ell$ and $k$, specifically for the case where there is only one distinguished candidate.

Let $I=(G, \kappa)$ be an instance of the {\prob{$\kappa$-Clique}} problem, where~$G=(\vset, \eset)$ is a~$d$-regular graph with~$d>0$. Let~$\m$ and~$\n$ denote the number of edges and the number of vertices of~$G$, respectively. We assume that both~$\kappa$ and~$n$ are divisible by~$6$. This assumption does not affect the {\wahns} of the problem.\footnote{If this condition is not met, we can transform $(G=(U, A), \kappa)$ into an equivalent instance $(G', 6\kappa)$ that satisfies the assumption in polynomial time as follows: For each vertex $u\in U$, we create a set~$B(u)$ of six vertices and let them form a clique in~$G'$. For each edge $\edge{u}{u'}$ in~$G$, we add all missing edges between~$B(u)$ and~$B(u')$, ensuring that $B(u)\cup B(u')$ forms a clique in~$G'$. It is evident that $(G, \kappa)$ and $(G',6\kappa)$ are equivalent.} Additionally, we assume that $\kappa^2+2d<2m$; otherwise, the problem can be solved in {\fpt}-time with respect to~$\kappa$.\footnote{If $2m\leq \kappa^2+2d$, given that $2m=n\cdot d$ and $d\geq 1$, we obtain $n\leq \kappa^2+2$. In this case, the instance can be solved in {\fpt}-time with respect to~$\kappa$ by enumerating all $\kappa$-subsets of vertices and checking whether any of them form a clique.}

We construct an instance $g(I)=((C, V), \discset, \ell, k)$ of {\pappdel{{\sav}}} as follows. For each vertex~$u$ in~$G$, we create a corresponding candidate, denoted by the same symbol for notational brevity. In addition, we create a candidate~$p$ as the only distinguished candidate. Let $C=\vset\cup \{p\}$ and let $J=\{p\}$. We create two classes of votes. To this end, we first arbitrarily group vertices of~$G$ into~$n/3$ sets of equal size, denoted as~$U_1$,~$U_2$,~$\dots$,~$U_t$, where $t=n/3$. 
This can be done as~$n$ is divisible by~$6$. 
For each group~$U_i$, we create $\frac{2\m-\kappa^2-2d+2}{2}$ votes, with each vote approving exactly the candidates in~$U_i$. Let~$V(U_i)$ denote the multiset of the votes created for~$U_i$, $i\in [t]$. As~$\kappa$ is even and $2m>\kappa^2+2d$, $\frac{2\m-\kappa^2-2d+2}{2}$ is a positive integer. Let $V(U)=\bigcup_{i\in [t]}V(U_i)$ be the multiset of all votes created for these sets. The second class of votes corresponds to edges of~$G$: for each edge~$\edge{u}{u'}\in A$, we create one vote $v(\edge{u}{u'})=\{u, u', p\}$. Figure~\ref{fig-appdel-sav-wah} illustrates the construction of the candidates and the votes. The construction is completed by setting $k=\kappa$ and $\ell=\frac{k\cdot (k-1)}{2}$. 

\begin{figure}
    \centering
    \includegraphics[width=\linewidth]{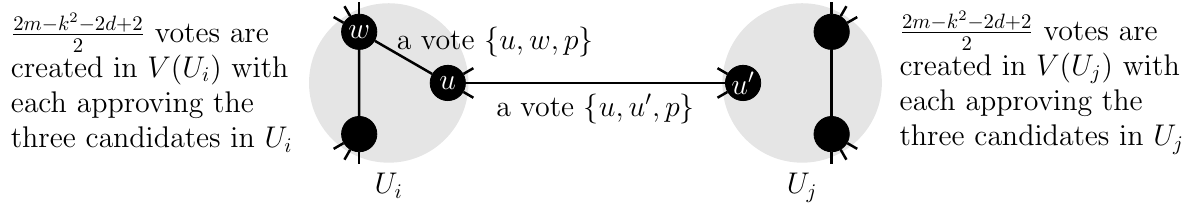}
    \caption{An illustration of the reduction presented in the proof of Theorem~\ref{thm-appdel-sav-nsav-wah-l-k}. For clarity, only the vertices in two groups, $U_i$ and $U_j$, are depicted. Each edge in $G$ is either within the same group, such as the edge between $w$ and $u$, or spans two different groups, like the edge between $u$ and $u'$. For every edge, there is a corresponding vote that approves exactly its two endpoints along with the distinguished candidate $p$. To simplify the illustration, only the votes corresponding to the edges $\edge{w}{u}$ and $\edge{u}{u'}$ are shown.}
    \label{fig-appdel-sav-wah}
\end{figure}

The reduction can be completed in polynomial time. It remains to show that~$I$ is a {\yesins} {\iff}~$g(I)$ is a {\yesins}. Notice that in the election~$(C, V)$, the {\sav} score of~$p$ is~$\frac{\m}{3}$, and that of every $u\in U$ is \[\frac{2\m-k^2-2d+2}{2}\times\frac{1}{3}+\frac{d}{3}=\frac{2\m-k^2+2}{6}.\] As~$p$ has the maximum {\sav} score,~$p$ is included in all winning $k$-committees of $(C, V)$. 

$(\Rightarrow)$ Assume that~$G$ contains a clique~$K$ of $\kappa=k$ vertices. Let~$A'$ be set of the edges whose both endpoints are within~$K$. Clearly,~$A'$ consists of exactly $\frac{k\cdot (k-1)}{2}=\ell$ edges. For each vote corresponding to an edge $\edge{u}{u'}\in A'$, we remove~$p$ from this vote. In total, we performed $\abs{A'}=\ell$ AppDel operations. 
Dropping~$p$ from a vote $v(\edge{u}{u'})$, where $\edge{u}{u'}$ is an edge in~$G$, decreases the {\sav} score of~$p$ by~$\frac{1}{3}$. As a result, the {\sav} score of~$p$ decreases to \[\frac{m}{3}-\frac{\abs{A'}}{3}=\frac{2\m-k\cdot (k-1)}{6}=\frac{2\m-k^2+k}{6}.\] As~$K$ is a clique of size $\kappa=k$, each vertex $u\in K$ is in exactly $k-1$ edges in~$A'$. Dropping~$p$ from each vote corresponding to an edge containing~$u$ increases the {\sav} score of~$u$ by $\frac{1}{2}-\frac{1}{3}=\frac{1}{6}$. Dropping~$p$ from any vote corresponding to an edge not containing~$u$ does not change the {\sav} score of~$u$. Therefore, after the above modifications, the {\sav} score of every~$u\in K$ increases to \[\frac{2\m-k^2+2}{6}+\frac{k-1}{6}=\frac{2\m-k^2+k+1}{6},\] which is strictly greater than the {\sav} score of~$p$. From $\abs{K}=k$, we know that~$p$ cannot be in any winning $k$-committees after the above modifications.

$(\Leftarrow)$ Assume that~$g(I)$ is a {\yesins}. 
We first prove the following claim.

\begin{claim} 
\label{claim-a}
There is a feasible solution where none of the votes in~$V(U)$ is changed.
\end{claim}

\begin{proof}
For the sake of contradiction, assume that in a feasible solution, we changed a vote $v\in V(U_i)$ for some $i\in [t]$. After the change, there are three possible cases: 
(1) The vote approves two candidates. 
(2) The vote approves only one candidate. 
(3) The vote does not approve any candidate. 
Note that changing any vote corresponding to an edge by removing~$p$ from the vote decreases the score gap between~$p$ and every other candidate by~$\frac{1}{3}$. However, changing~$v$ as in Case~(1) decreases the score gap between~$p$ and the two candidates not removed from $v$ by~$\frac{1}{6}$, while strictly increasing the score gap between $p$ and the candidate removed from $v$. Changing~$v$ as in Case~(3) does not decrease the score gap between~$p$ and any other candidate; in fact, it even increases the score gap between~$p$ and the candidates in~$v$. Therefore, if the vote~$v$ is changed as in Cases (1) or (3), we can obtain a new feasible solution by leaving~$v$ unchanged and instead changing a vote corresponding to some edge that has not yet been changed. Such a vote must exist because, by our initial assumption, it holds that $\ell=\frac{\kappa\cdot (\kappa-1)}{2}< \frac{\kappa^2+2d}{2}<m$. In Case~(2), the score gap between~$p$ and every other candidate decreases by at most~$\frac{2}{3}$ due to two deletion operations. In this case, we can obtain a new feasible solution by keeping~$v$ unchanged and instead changing two edge-votes by removing~$p$ that have not been changed in the feasible solution. By a similar argument as above, such two votes must exist.  
\end{proof}

By Claim~\ref{claim-a}, we may assume that in a feasible solution, all modified votes are edge-votes. Similar to the proof of Claim~\ref{claim-a}, after changing a vote, there are three possibilities. However, changing a vote into an empty vote is not meaningful, as this does not decrease the score gap between 
$p$ and any other candidate. Therefore, we consider only the first two cases. Additionally, it is optimal to fully utilize the allowed number of deletion operations. Thus, we assume that in the feasible solution, exactly $\ell=\frac{\kappa \cdot (\kappa-1)}{2}$ AppDel operations are performed. 
Moreover, if an edge-vote is selected to be changed, candidate~$p$ should be among the removed candidates. 
Let~$t_1$ and~$t_2$ denote the number of votes that approve exactly one candidate and exactly two candidates after the changes, respectively. Then, we have
\[2t_1+t_2=\frac{\kappa \cdot (\kappa-1)}{2}.\]
Since~$p$ is removed in all changed votes, its {\sav} score decreases by $\frac{t_1+t_2}{3}$, resulting in a final {\sav} score of 
\[\frac{m}{3}-\frac{t_1+t_2}{3}=\frac{m+t_1}{3}-\frac{\kappa \cdot (\kappa-1)}{6}.\] 
Removing a candidate from a vote effectively redistributes the score that the candidate would have received to the remaining approved candidates in the vote.  
Let~$K$ be the set of candidates whose {\sav} scores increase after the changes in the feasible solution. The total increase in the {\sav} scores of the candidates in~$K$ is at most
\[\frac{2t_1}{3}+\frac{t_2}{3}=\frac{\kappa \cdot (\kappa-1)}{6}.\] 
If~$p$ is not included in any winning $k$-committee after the changes in the feasible solution, then there must exist a subset $K'\subseteq K$ of at least~$k=\kappa$ candidates whose final {\sav} scores are strictly greater than the final {\sav} score of~$p$. Since the total increase in the {\sav} scores of these candidates is at most $\frac{\kappa \cdot (\kappa-1)}{6}$, it follows that among the candidates in~$K'$, the one with the lowest {\sav} score experiences an increase of at most $\frac{\kappa-1}{6}$, leading to a final {\sav} score of at most 
\[\frac{2\m-k^2+2}{6}+\frac{k-1}{6}=\frac{2\m-k^2+k+1}{6}.\] 
As~$p$ has a strictly lower {\sav} score than this candidate, we obtain 
\[\frac{m+t_1}{3}-\frac{\kappa \cdot (\kappa-1)}{6}<\frac{2\m-k^2+k+1}{6}.\] 
This implies that $t_1=0$, and thus $t_2=\frac{\kappa\cdot (\kappa-1)}{2}=\ell$. 
Thus, the final {\sav} score of~$p$ is $\frac{m}{3}-\frac{\kappa\cdot (\kappa-1)}{6}$. This further implies that changing~$\ell$ edge-votes in a feasible solution increases the {\sav} score of each candidate in~$K'$ by at least $\frac{\kappa\cdot (\kappa-1)}{6}$. By the construction above, and since changing an edge-vote $v(\edge{u}{u'})$ increases the {\sav} scores of both~$u$ and~$u'$ by exactly $\frac{1}{2}-\frac{1}{3}=\frac{1}{6}$, we conclude that each candidate in~$K'$ corresponds to a vertex incident to at least $\kappa-1$ edges whose corresponding votes are changed. Since exactly $\ell=\frac{\kappa \cdot (\kappa-1)}{2}$ edge-votes are changed in the feasible solution, it follows that the set of vertices corresponding to~$K'$ forms a clique in the graph~$G$.
\end{proof}

\hide{
\begin{theorem}
{\pappdel{{\nsav}}}  is {\nph} even if $k=1$ and $r=3$.
\end{theorem}

\begin{proof}
We partition~$C(\xs)$ into three subsets~$C^1(\xs)$,~$C^2(\xs)$, and~$C^3(\xs)$ such that each subset consists of exactly~$\xsize$ candidates. We create the following votes. For each of~$C^i(\xs)$ where $i\in [3]$, we create $\xsize^2+\xsize$ votes each of which approves exactly the candidates in~$C^i(\xs)$. For each $\xce\in \xc$, we create one vote~$v(\xce)$ which approves~$\p$ and the three candidates corresponding to the elements in~$\xce$. Finally, we set $\ell=3\xsize$.

The only optimal solution is to select~$\xsize$ votes  corresponding to an exact $3$-set cover $\xc'\subseteq \xc$ and for each of them remove the three candidates corresponding to the three elements in~$\xce$.
\end{proof}
}

\subsection{Vote-Level Change}
From this section, we study vote-level operations. We show that problems associated with these operations are computationally hard even in several special cases.


%

\begin{theorem}
\label{thm-pvc-av-nph}
{\emph{\pvc{\av}}} for all possible values of $r\geq 4$ are {\emph\nph} even when $k=1$.
\end{theorem}

\begin{proof}
We prove the theorem via a reduction from the {\prob{RX3C}} problem. Let~$r\geq 4$ be an integer. Let $(\xs, \xc)$ be an instance of {\prob{RX3C}}, where $\abs{\xs}=\abs{\xc}=3\xsize$. Without loss of generality, we assume that $\xsize>3$. We create an instance $((C, V), \discset, \ell, k, r)$ of {\pvc{\av}} as follows. The candidate set is $C=\xs\cup \{p\}$, consisting of $3\kappa+1$ candidates. Let $\discset=\xs$. Now we construct the votes. For each $\xse\in \xs$, we create $\xsize-3$ votes, each approving only~$\xse$. For each $\xce\in \xc$, we create one vote~$v(\xce)$ approving exactly the three candidates in~$\xce$, i.e., $v(\xce)=\xce$. We define~$V$ as the multiset of the above created votes. Let $\ell=\xsize$ and let $k=1$. It is easy to see that~$p$ has an {\av} score of~$0$, while each $\xse\in \xs$ has an {\av} score of~$(\xsize-3)+3=\xsize$ in~$(C, V)$.

The reduction can be completed in polynomial time. We show below that the {\prob{RX3C}} instance is a {\yesins} {\iff} the constructed {\pvc{\av}} instance is a {\yesins}.

$(\Rightarrow)$ Assume that there is an exact set cover $\xc'\subseteq \xc$ of~$\xs$. We change all the~$\xsize$ votes corresponding to the~$\xsize$ elements of~$\xc'$ so that, after the change, all of them approve only~$p$. Clearly, the Hamming distance between each modified vote and its original is four. Moreover, in the new election,~$p$ receives~$\xsize$ approvals. As~$\xc'$ is an exact set cover, for every candidate $\xse\in\xs$, there is exactly one vote that originally approves~$\xse$ but no longer does after the changes. Therefore, in the new election, every candidate~$\xse\in A$ receives $\xsize-1$ approvals, implying that none of the candidates in~$\discset$ can be a winner.

$(\Leftarrow)$ Assume that we can change at most~$\ell$ votes without exceeding the distance bound~$r$, ensuring that none of the candidates in~$\discset$ is a winner of the resulting election. This is equivalent to making~$p$ the unique winner after the changes. First, observe that if a vote is determined to be changed, it is always optimal to change it into a vote that approves only~$p$. Since $r\geq 4$ and each vote originally approves at most three candidates, any vote can be changed to~$\{p\}$ without violating the distance restriction. Second, if the constructed instance is a {\yesins}, there exists at least one feasible solution where exactly~$\ell$ votes are changed%
.
Given such a feasible solution, and based on the two observations, after changing~$\ell$ votes in the constructed instance as described above, the {\av} score of~$p$ becomes exactly $\ell=\xsize$. Consequently, for every $\xse\in \xs$, there is at least one vote~$v(\xce)$, $\xce\in \xc$, that originally approved~$\xse$ but is changed into a vote that only approves~$p$. By construction, a vote~$v(\xce)$, $\xce\in \xc$, originally approved~$\xse$ if and only if $\xse\in \xce$. Since this holds for all $\xse\in \xs$, the subcollection corresponding to the solution is an exact set cover of~$\xs$.
\end{proof}

In the proof of Theorem~\ref{thm-pvc-av-nph}, the number of distinguished candidates is not bounded by a constant. This naturally prompts the question of whether the problem is {\fpt} with respect to the number of distinguished candidates or the combined parameter~$k+\abs{J}$. The following result provides a negative answer to the question.

\begin{theorem}
\label{thm-vc-av-wa-hard-k-ell}
{\emph{\pvc{\av}}} for all possible values of $r\geq 3$ is {\emph{\wah}} with respect to the parameters $\ell+k$. This holds even when $\abs{\discset}=1$.
\end{theorem}

\begin{proof}
We prove the theorem via a reduction from the {\prob{$\kappa$-Clique}} problem restricted to regular graphs. Let $(G,\kappa)$ be an instance of the  {\prob{$\kappa$-Clique}}  problem, where~$G=(U, A)$ is a~$\dd$ regular graph with~$\dd>0$. Let~$m$ be the number of edges in~$G$. Without loss of generality, we assume that~$\kappa\geq 2$. Moreover, we assume that $d>{\kappa}^3$; otherwise, the instance can be solved in {\fpt}-time with respect to~$\kappa$. Let~$r$ be an integer at least~$3$. We construct an {\pvc{\av}} instance $((C, V), \discset, \ell, k, r)$ as follows. 
The candidate set is $C=U\cup \{p\}$, where~$p\not\in U$ is the only distinguished candidate, i.e., $\discset=\{p\}$. Regarding the votes, we first create~$m$ votes corresponding to the edges in~$G$. Particularly, for each edge~$\edge{u}{u'}\in A$, we create a vote approving all candidates in~$C$ except~$u$ and~$u'$. Then, we create $d+1-\frac{(\kappa-1)\cdot (\kappa+2)}{2}$ votes approving all candidates in~$C$ except~$p$. Note that under our assumption $d>\kappa^3$ and $\kappa\geq 2$, $d+1-\frac{(\kappa-1)\cdot (\kappa+2)}{2}$ is a positive integer. We complete the construction by setting $\ell=\frac{\kappa\cdot (\kappa-1)}{2}$ and $k=\kappa$, i.e., we are allowed to change at most $\frac{\kappa\cdot (\kappa-1)}{2}$ votes, and we aim to select exactly~$\kappa$ winners.

Clearly, the above construction can be implemented in polynomial time. It remains to analyze the correctness of the reduction. Let us first consider the {\av} scores of the candidates. The {\av} score of~$\p$ is
\[\scoreE{\p}{E}{\av}=m,\]
and the {\av} score of every other candidate $u\in U$ is
\begin{align*}
\scoreE{u}{E}{\av} & =\left(d+1-\frac{(\kappa-1)\cdot (\kappa+2)}{2}\right)+\left(m-d\right)\\
             & =1+m-\frac{(\kappa-1) \cdot (\kappa+2)}{2}.\\
\end{align*}
Under the assumption $\kappa\geq 2$, it holds that $\scoreE{u}{E}{\av} < \scoreE{p}{E}{\av}$ for every candidate $u\in U$.

$(\Rightarrow)$ Assume that there is a clique $U'\subseteq U$ of~$\kappa$ vertices in the graph~$G$. We modify all votes corresponding to edges between vertices in the clique~$U'$, i.e., all edges in~$G[U']$. Clearly, there are exactly $\frac{\kappa\cdot (\kappa-1)}{2}=\ell$ such votes. We modify them so that all candidates except~$\p$ are approved in these votes. The distance between each new vote and its original vote is exactly~$3$. Modifying a vote corresponding to some edge~$\edge{u}{u'}$ this way decreases the {\av} score of~$\p$ by one, while increases the {\av} scores of both~$u$ and~$u'$ by one. So, after all these modifications, the {\av} score of~$\p$ becomes
\[m-\frac{\kappa\cdot (\kappa-1)}{2}.\]
For each $u\in U'$, as there are exactly $\kappa-1$ edges incident to~$u$ in~$G[U']$, after modifying the above~$\ell$ votes, the {\av} score of~$u$ becomes
\[\scoreE{u}{E}{\av}+\kappa-1=1+m-\frac{\kappa \cdot (\kappa-1)}{2},\]
strictly greater than the final score of~$\p$, implying that~$\p$ cannot be in any winning $k$-committees after all the modifications.

$(\Leftarrow)$ Assume that we can change at most $\ell=\frac{\kappa \cdot (\kappa-1)}{2}$ votes in~$(C, V)$ so that~$\p$ is not included in any winning $k$-committee of the resulting election. Observe that it is always better to modify votes corresponding to the edges than those who already disapprove~$\p$ and approve all the other candidates. Moreover, it is optimal to fully use the number of operations allowed, and when a vote corresponding to an edge is determined to be modified, it is optimal to change it to a vote that approves all candidates except the distinguished candidate~$\p$. The final {\av} score of~$\p$ is then determined:
$\scoreE{p}{E}{\av}-\ell=m-\frac{\kappa \cdot (\kappa-1)}{2}$.
This implies that after the changes, there must be at least~$k$ candidates whose {\av} scores increase by at least $\kappa-1$ each.  Note that when a vote corresponding to some edge~$\edge{u}{u'}$ is changed, only the {\av} scores of~$u$ and~$u'$ increase, each by one.
Let~$U'$ be the set of candidates whose scores are increased after the changes of the votes in a feasible solution.
Therefore,~$U'$ consists of the vertices spanned by all edges whose corresponding votes are changed.
By the above analysis, changing $\ell=\frac{\kappa\cdot (\kappa-1)}{2}$ edge-votes can increase the total {\av} scores of candidates in~$U'$ by at most $\kappa\cdot (\kappa-1)$.
This directly implies that~$U'$ consists of exactly~$k$ candidates. From the fact that  $\frac{\kappa\cdot (\kappa-1)}{2}$ edges span exactly~$\kappa$ vertices if and only if the set of spanned vertices is a clique, we know that~$U'$ is a clique in~$G$, and hence the given {\prob{$\kappa$-Clique}} instance is a {\yesins}.
\end{proof}

Now, we move on to SAV and NSAV. 

\begin{theorem}
\label{thm-vc-sav-nsav-nph-r-4}
{\emph{\pvc{{\sav}}}}  and {\emph{\pvc{{\nsav}}}} for all integers $r\geq 4$ are {\emph\nph}. This holds even if $k=1$.
\end{theorem}

\begin{proof}
We give only the detailed proof for {\sav}. The proof for {\nsav} can be obtained by utilizing Lemma~\ref{lem-relation-sav-nsav}.

The proof follows a similar structure to the proof of Theorem~\ref{thm-pvc-av-nph} and is based on a reduction from the {\prob{RX3C}} problem. Let $r\geq 4$ be an integer. Let $I=(\xs, \xc)$ be an instance of the {\prob{RX3C}} problem with $\abs{\xs}=\abs{\xc}=3\xsize>0$. We create an instance $g(I)=((C, V), \discset, \ell, k, r)$ of {\pvc{{\sav}}} as follows. The candidate set is $C=\xs\cup \{p\}$ with $\discset=\xs$ being the set of distinguished candidates. 
We have in total $3\xsize+1$ candidates. We create the following votes.
For each $\xse\in \xs$, we create $\xsize-1$ votes in~$V$, each approving only~$\xse$. In addition, for each $\xce\in \xc$, we create one vote $v(\xce)=\xce$ in~$V$ which approves exactly the three candidates corresponding to the elements in~$\xce$. We have in total $3\xsize \cdot (\xsize-1)+3\xsize=3\xsize^2$ votes. The reduction is completed by setting $\ell=\xsize$ and $k=1$. 
In the following, we show that~$I$ is a {\yesins} if and only if~$g(I)$ is a {\yesins}. Note that in the election~$(C, V)$, the {\sav} score of~$p$ is~$0$, and that of every $\xse\in \xs$ is $(\xsize-1)+3\times \frac{1}{3}=\xsize$. 

($\Rightarrow$) Assume that there is an exact set cover $\xc'\subseteq \xc$ of~$\xs$. We change the~$\xsize$ votes corresponding to~$\xc'$ so that they approve only the candidate~$p$ after the changes. The Hamming distance between each new vote and its original vote is four, which is bounded from above by~$r\geq 4$. In the new election, the {\sav} score of~$p$ is~$\xsize$, and that of every candidate $\xse\in \xs$ is $\xsize-\frac{1}{3}$, implying that none of~$\discset$ is a winner of the new election.

($\Leftarrow$) Suppose that we can change at most $\ell=\xsize$ votes by performing at most~$\ell$ vote-level change operations to exclude any of~$\discset$ from any winning $k$-committee. We observe that if a vote is determined to be changed, it is always optimal to change it to a vote that exclusively approves the candidate~$p$. As a consequence, we may assume that~$p$ has a {\sav} score of~$\xsize$ in the final election. Therefore, for every candidate $\xse\in\xs$, at least one vote approving~$\xse$ is changed. As we change at most $\ell=\xsize$ votes and $\abs{\xs}=3\xsize$, all changed votes must be from those corresponding to~$\xc$. As a vote $v(\xce)$ approves~$\xse$ if and only if $\xse\in \xce$, it holds that the subcollection corresponding to the changed votes is set cover of~$\xs$, and as we can change at most $\ell=\xsize$ votes, it is in fact an exact set cover of $\xs$.
\end{proof}

Regarding fixed-parameter intractability results, we have the following theorem.

\begin{theorem}
\label{thm-pvc-sav-nsav-wah-l-k-r-1}
For all integers $r\geq 1$, {\emph{\pvc{{\sav}}}}  and {\emph{\pvc{{\nsav}}}}  are {\emph\wah} with respect to~$\ell+k$, even when $\abs{\discset}=1$.
\end{theorem}

\begin{proof}
We provide only the proof for {\sav}. The proof for {\nsav} can be obtained from the following proof by introducing a number of dummy candidates based on Lemma~\ref{lem-relation-sav-nsav}. 

Our proof is based on a reduction from the {\prob{$\kappa$-Clique}} problem on regular graphs. Let~$r$ be a positive integer. Let $I=(G, \kappa)$ be an instance of the {\prob{$\kappa$-Clique}} problem, where~$G=(U, A)$ is a $\dd$-regular graph with $\dd>0$. We assume that~$\kappa$ is odd. This assumption does not change the {\wahns} of the problem.\footnote{To show the {\wahns} of the {\prob{$\kappa$-Clique}} problem on regular graphs, Cai~\shortcite{DBLP:journals/cj/Cai08} presented a polynomial-time reduction that transforms an instance of the {\prob{$\kappa$-Clique}} problem in general graphs into an equivalent instance where the input graph is regular. Notably, the parameters in both instances remain unchanged. Thus, to establish that the assumption does not affect the {\wahns} of the problem when restricted to regular graphs, it suffices to show that the assumption does not alter the {\wahns} of the problem in general. To achieve this, consider an instance of the {\prob{$\kappa$-Clique}} problem where~$\kappa$ is even. We modify the graph by adding a new vertex  adjacent to all existing vertices. Clearly, the original instance contains a clique of~$\kappa$ vertices if and only if the modified graph contains a clique of~$\kappa+1$ vertices.} Let~$m=\abs{A}$ denote the number of edges in~$G$. In addition, we assume that $m>d+\kappa^2+\kappa\cdot r$ (note that~$r$ is an integer constant), $d\geq \frac{\kappa \cdot (\kappa-1)}{2}$, and $\kappa\geq 5$, which does not change the {\wahns} of the problem too. (If they are not satisfied the problem can be solved in {\fpt}-time with respect to~$\kappa$~\cite{DBLP:journals/jea/EppsteinLS13}.) Now we create an instance $g(I)=((C, V), \discset, \ell, k, r)$ of {\pvdc{{\sav}}} as follows. 
 Let  
    \[
    t = m - d - \frac{\kappa \cdot (\kappa - 3)}{2} - \frac{\kappa - 1}{2} \cdot (r+2).
    \]  
    By our assumptions,~$t$ is a positive integer. 
 
 We begin by creating a candidate~$p$ and defining $\discset = \{p\}$, meaning that $p$ is the only distinguished candidate.   
    Then, for each vertex $u \in U$, we create a candidate~$c(u)$ and introduce~$t$ subsets $C(u,1), C(u,2), \dots, C(u,t)$ of dummy candidates, where each subset consists of exactly $r+1$ candidates.  
    Additionally, for each edge $\edge{u}{u'} \in A$, we create a set~$C(\edge{u}{u'})$ of $r-1$ dummy candidates.  

  
  Now, we define the votes in~$V$. For each vertex $u \in U$, we introduce~$t$ votes, denoted as $v(u,1)$, $v(u,2)$, $\dots$, $v(u,t)$, where each vote is given by  
    \[
    v(u,i) = \{c(u)\} \cup C(u,i),
    \]  
    for all $i \in [t]$.  
    Furthermore, for each edge $\edge{u}{u'}$ in~$G$, we create a vote  
    \[
    v(\edge{u}{u'}) = \{c(u), c(u'), p\} \cup C(\edge{u}{u'}).
    \]  
    Let $V(A) = \{v(\edge{u}{u'}) \setmid \edge{u}{u'} \in A\}$ be the subset of votes corresponding to edges of $G$.  

\begin{figure}
    \centering
    \includegraphics[width=0.85\linewidth]{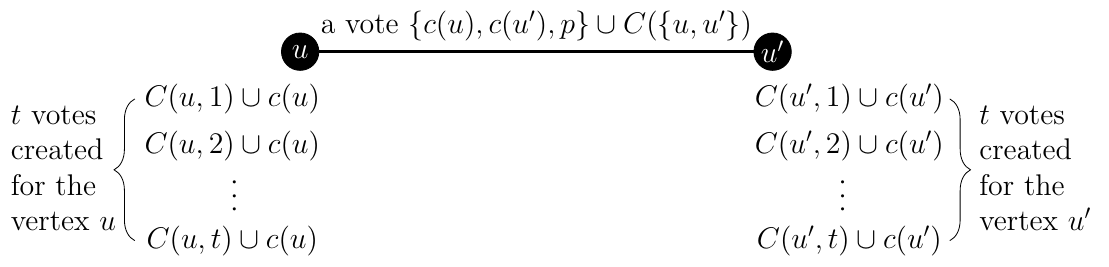}
    \caption{An illustration of the election construction in the proof of Theorem~\ref{thm-pvc-sav-nsav-wah-l-k-r-1}. Each set  $C(u, i)$ consists of exactly $r+1$ candidates, while each set $C(\edge{u}{u'})$, created for an edge $\edge{u}{u'}$, consists of exactly $r-1$ candidates.} 
    \label{fig-vote-change-sav-wa-hard}
\end{figure}
Refer to Figure~\ref{fig-vote-change-sav-wa-hard} for an illustration of the construction of $(C, V)$. 
The reduction is completed by setting $\ell=\frac{\kappa \cdot (\kappa-1)}{2}$ and $k=\kappa$. It is clear that $g(I)$ can be constructed in polynomial time. 
Broadly speaking, the instance~$g(I)$ is designed so that all dummy candidates have no possibility of becoming winners. Furthermore, we can modify the votes corresponding to edges within a clique of size~$\kappa$ by removing dummy candidates and~$p$, ensuring that the candidates representing the clique achieve strictly higher {\sav} scores than  after these modifications. 
The formal proof of the reduction is as follows.

Let $E=(C, V)$. Before the formal proof, it is helpful to analyze the {\sav} scores of the candidates in the election~$E$. Observe that each vote approves exactly $r+2$ candidates. As each dummy candidate is approved by exactly one vote, the {\sav} score of each dummy candidate is $\frac{1}{r+2}$. The {\sav} score of~$p$ is~$\frac{m}{r+2}$, and the {\sav} score of each~$c(u)$ where $u\in U$ is $\frac{d+t}{r+2}$.  

($\Rightarrow$) Assume that there is a clique $K\subseteq U$ of~$\kappa$ vertices in~$G$. Let $A'=\{\edge{u}{u'}\in A \setmid \{u, u'\}\subseteq K\}$ be the set of edges in the subgraph of $G$ induced by~$K$. It follows that $\abs{A'}=\frac{\kappa \cdot (\kappa-1)}{2}$. We change the $\ell=\frac{\kappa \cdot (\kappa-1)}{2}$ votes corresponding to~$A'$ by removing~$p$ and all dummy candidates from these votes. Precisely, for each $\edge{u}{u'}\in A'$, we change the vote $v(\edge{u}{u'})$ into~$\{c(u), c(u')\}$. The Hamming distance between the original and the new votes is~$r$. Let~$E'$ denote the resulting election.  
Now we analyze the {\sav} scores of the candidates in~$E'$. As changing each vote as above decreases the {\sav} score of~$p$ by~$\frac{1}{r+2}$, the {\sav} score of~$p$ in~$E'$ is 
\[\scoreE{p}{E'}{\sav}=\frac{m-(\kappa\cdot (\kappa-1)/{2})}{r+2}.\] The {\sav} score of each dummy candidate in~$E'$ is at most $\frac{1}{r+2}$. As~$K$ is a clique, each $u\in K$ is incident to exactly $\kappa-1$ edges in~$A'$. From the above construction, it follows that there are exactly~$\kappa-1$ votes approving~$c(u)$ being changed. As changing each vote approving~$c(u)$ as above increases the {\sav} score of~$c(u)$ by $\frac{1}{2}-\frac{1}{r+2}$, the {\sav} score of~$c(u)$ in~$E'$ climbs to
\[\frac{d+t}{r+2}+(\kappa-1)\cdot \left(\frac{1}{2}-\frac{1}{r+2}\right) = \scoreE{p}{E'}{\sav}+\frac{1}{r+2}.\] As $k=\xsize=\abs{K}$, this implies that~$p$ cannot be in any {\sav} winning $k$-committees of~$E'$.

($\Leftarrow$) Assume that~$g(I)$ is a {\yesins}. That is, there exists $V'\subseteq V$ of at most~$\ell$ votes that can be changed into votes at a Hamming distance of at most~$r$ from their original votes, ensuring that none of the candidates in~$\discset$ is contained in any wining $k$-committee of the resulting election. Observe that if a vote created for a vertex $u\in U$ is included in~$V'$, we can replace this vote in~$V'$ with a vote corresponding to an edge incident to~$u$ that is not currently  in~$V'$. Such a vote must exist, given that $d\geq \frac{\kappa \cdot (\kappa-1)}{2}$ and $\ell=\frac{\kappa \cdot (\kappa-1)}{2}$.\footnote{Changing a vote (subject to the distance bound) created for a vertex $u\in U$ increases the score gap between~$c(u)$ and $p$ by at most $1/2-1/(r+2)$. However, changing a vote $v(\edge{u}{u'})$ by removing the dummy candidates and $p$ not only increases the score gap between~$c(u)$ and~$p$ but also amplifies the score gap between~$c(u')$ and~$p$ by at least the same amount.} In addition, if a vote $v(\edge{u}{u'})$ corresponding to an edge~$\edge{u}{u'}$ is contained in~$V'$, removing the~$r-1$ dummy candidates and the distinguished candidate~$p$ is an optimal strategy. Altering a vote in this manner minimizes the {\sav} score of candidate~$p$ to the greatest extent possible, while simultaneously maximizing the scores of nondummy candidates. Moreover, it is always better to fully use the budget, i.e., changing exactly~$\ell$ votes than changing less than~$\ell$ votes. In light of these observations, we may assume that  $\abs{V'}=\ell=\frac{\kappa\cdot (\kappa-1)}{2}$, and all votes in $V'$ are from~$V(A)$. Defining~$E'$ as the election obtained from $E$ by changing all votes from $V'$ by removing the dummy candidates and~$p$, we know then that~$p$ is not contained in any winning $k$-committee of~$E'$. The {\sav} score of~$p$ in~$E'$ is $\frac{m-(\kappa\cdot (\kappa-1)/{2})}{r+2}$, as analyzed in the above proof for the $\Rightarrow$ direction. Let~$K$ be the set of vertices in the edges corresponding to the votes in~$V'$. Clearly, we have $\abs{K}\geq \kappa$. Furthermore, equality holds if and only if $K$ forms a clique in~$G$. As changing one vote $v(\edge{u}{u'})$ where $\edge{u}{u'}\in A$ in a way as discussed above in effect shifts the total score $\frac{r}{r+2}$ of $C(\edge{u}{u'})\cup \{p\}$ received from this vote to the two candidates~$c(u)$ and~$c(u')$, the total increase in the {\sav} score of candidates corresponding to~$K$ is $\frac{r}{r+2}\cdot \ell=\frac{r}{r+2} \cdot \frac{\kappa \cdot (\kappa-1)}{2}$. 
As~$p$ is not in any winning $k$-committees of~$E'$, we know that there are at least~$k$ candidates corresponding to vertices in~$K$ each of whose {\sav} score increases by at least 
\[\frac{m-(\kappa\cdot (\kappa-1)/{2})}{r+2}+\frac{1}{r+2}-\frac{d+t}{r+2}=\frac{r\cdot (\kappa-1)}{2(r+2)}\] after the changes. (On the left side, the first term represents the {\sav} score of~$p$ in the final election~$E'$, the third term corresponds to the {\sav} score of every vertex-candidate in the original election $E$, and the second term provides a lower bound on the minimum possible positive increase in the {\sav} score due to the vote-change operations.) It is then evident that the total increase in the {\sav} scores of the $k$ candidates is exactly the total increase in the {\sav} score of candidates corresponding to~$K$. Consequently, $\abs{K}=\kappa$ holds, and moreover,~$K$ forms a clique in~$G$. 

To prove the hardness for NSAV, recall that in the original election $E$, the SAV scores of candidates are $\frac{m}{r+2}$, $\frac{d+t}{r+2}$, and $\frac{1}{r+2}$. The minimum positive difference between the three values is bounded below by~$\frac{1}{r+2}$. As the total number of candidates is strictly larger than $r+2$, we can add a sufficiently large number of dummy candidates, as suggested in Lemma~\ref{lem-relation-sav-nsav}, to the above reduction to establish the {\wahns} of {\pvc{\nsav}} as stated in the theorem.
\end{proof}

\subsection{Vote-Level Addition Change}

In the previous section, we showed that {\pvc{\av}} is already {\nph} even when $k=1$ (Theorem~\ref{thm-pvc-av-nph}). However, this is not the case for the vote-level addition operation.

\begin{theorem}
\label{thm-vac-av-poly-k-1}
For all possible values of~$r$, {\emph{\pvac{\av}}}  is polynomial-time solvable if $k=1$.
\end{theorem}

\begin{proof}
Let $(E, J, k, \ell, r)$ be an instance of {\pvac{\av}}, where $E=(C, V)$, $J\subseteq C$,  $k=1$, and $r$ is a nonnegative integer. 
If $r=0$, we solve the instance by checking whether there exists a nondistinguished candidate with an {\av} score strictly greater than that of any distinguished candidates, which can be done in polynomial time. We now assume that $r\geq 1$.
Let~$s$ be the maximum {\av} score of the distinguished candidates, i.e., for all $c\in J$, it holds that $\scoreE{c}{E}{\av} \leq s$, and there exists $c'\in J$ with $\scoreE{c'}{E}{\av}=s$. Let $c^{\star}\in C\setminus J$ be a candidate with the highest {\av} score among candidates in $C\setminus J$. Define~$\overline{V}(c^{\star})$ as the submultiset of votes in~$V$ that disapprove~$c^{\star}$. Our algorithm returns ``{\yes}'' if $\scoreE{c^{\star}}{E}{\av}+\min\{\ell,\abs{\overline{V}(c^{\star})}\}\geq s+1$, and returns ``{\no}'' otherwise. The correctness of the algorithm is straightforward.
\end{proof}

Theorem~\ref{thm-vac-av-poly-k-1} prompts the question of whether {{\pvac{\av}}} can be solved in {\fpt}-time with respect to~$k$. 
The subsequent theorem provides a negative and robust answer: an {\fpt}-algorithm is unlikely even when parameterized by $\ell+k$, and in the presence of only one distinguished candidate. 

\begin{theorem}
\label{thm-vac-av-wah}
For all integers $r\geq 2$, {\emph{\pvac{\av}}} is {\emph\wah} when parameterized by $\ell+k$. Moreover, this holds even when $\abs{\discset}=1$.
\end{theorem}

\begin{proof}
We prove the theorem via a reduction from the {\prob{$\kappa$-Clique}} problem restricted to regular graphs. Let $I=(G, \kappa)$ be an instance of the {\prob{$\kappa$-Clique}} problem, where~$G=(U, A)$ is a $d$-regular graph. We assume that $m>d>\kappa>2$, since otherwise the problem can be solved in {\fpt}-time with respect to~$\kappa$~\cite{DBLP:journals/jea/EppsteinLS13}. Let~$m$ be the number of edges in the graph~$G$. We create an instance~$g(I)=((C,V), \discset, \ell, k, r)$ of {\pvac{\av}} as follows. The candidate set is $C=U\cup \{p\}$. Let~$\discset=\{p\}$. Regarding the votes, we first create $d-\kappa+2$ votes each of which approves all candidates in~$C$ except the distinguished candidate~$p$. From the above assumption, it follows that $d-\kappa+2$ is a positive integer. In addition, for each edge~$\edge{u}{u'}\in A$, we create one vote $v(\edge{u}{u'})$ approving all candidates in~$C$ except~$u$ and~$u'$, i.e., $v(\edge{u}{u'})=(U\cup \{p\})\setminus \{u, u'\}$. We set~$k=\kappa$ and $\ell=\frac{\kappa\cdot (\kappa-1)}{2}$. It is fairly easy to check that in the election~$(C, V)$, the {\av} score of~$p$ is~$m$, and the {\av} score of every candidate $u\in U$ is $(d-\kappa+2)+(m-d)=m-\kappa+2$, which is a smaller than the {\av} score of~$p$. Hence,~$p$ must be included in all winning $k$-committees in the current election. The construction clearly can be done in polynomial time. In the following, we show that for all possible values of $r\geq 2$, the instance~$I$ is a {\yesins} if and only if~$g(I)$ is a {\yesins}.

($\Rightarrow$) Assume that there is a clique~$K\subseteq U$ of~$\xsize$ vertices in the graph~$G$. Let $A'=\{e\in A \setmid e\subseteq K\}$ be the set of edges whose both endpoints are in~$K$. As~$K$ is a clique, $\abs{A'}=\frac{\kappa\cdot(\kappa-1)}{2}=\ell$. We change the~$\ell$ votes corresponding to~$A'$ so that they approve all candidates after the changes. Clearly, the distance between a vote $v(\edge{u}{u'})$ and the new vote after changing $v(\edge{u}{u'})$ in the above way is exactly two. After the changes, the {\av} score of~$p$ remains unchanged. However, for every candidate $u\in K$, there are exactly $\kappa-1$ votes which originally do not approve~$u$ but are changed into ones approving all candidates including~$u$. Therefore, the {\av} score of every candidate from~$K$ after the above changes is $(m-\kappa+2)+(\kappa-1)=m+1$, implying that~$K$ is the unique winning $k$-committee. As $p\not\in K$, the instance~$g(I)$ is a {\yesins}.

($\Leftarrow$) Assume that we can change at most~$\ell$ votes so that~$p$ is not in any winning $k$-committee. Observe that there must be a feasible solution where all changed votes are from those corresponding to~$A$. Moreover, if a vote corresponding to an edge~$\edge{u}{u'}$ is determined to be changed, it is optimal to change it by adding the two candidates~$u$ and~$u'$ into the vote. Let~$A'$ be the set of edges corresponding to the votes that are changed in such a solution. In addition, let~$K$ denote the set of vertices that are incident to at least one edge in~$A'$. By the above observation, we know that in the final election, the {\av} scores of~$p$ and those in $U\setminus K$ remain unchanged. This implies that $\abs{K}\geq k$ and, moreover, there exists a subset $K'\subseteq K$ of~$k$ candidates such that, for every $u\in K'$, at least $\kappa-1$ votes which do not approve~$u$ originally are changed in the solution. By the above construction, this is equivalent that every candidate-vertex in~$K'$ is incident to at least $\kappa-1$ edges in~$A'$. Given that $\abs{A'}=\frac{\kappa\cdot (\kappa-1)}{2}$, this is possible only when~$K'=k$ and~$K'$ is a clique of~$G$.
\end{proof}

Unlike {\av}, for {\sav} and {\nsav}, we already have {\nphns} even when both~$k$ and~$r$ are equal to~$1$.

\begin{theorem}
\label{thm-vac-sav-nph-k-1-r-1}
For all $r\geq 1$, {\emph{\pvac{{\sav}}}} and {\emph{\pvac{{\nsav}}}} are {\emph\nph}, even when $k=1$.
\end{theorem}

\begin{proof}
We prove Theorem~\ref{thm-vac-sav-nph-k-1-r-1} for SAV via a reduction from the {\prob{RX3C}} problem.  Let $I=(\xs, \xc)$ be an {\prob{RX3C}} instance such that $|\xs|=\abs{\xc}=3\xsize>0$. Similar to the proof of Theorem~\ref{thm-appadd-sav-nsav-np-hard}, we assume that~$\xsize\geq 6$ and~$\xsize$ is even. Let $r\geq 1$ be an integer. We construct an instance $g(I)=((C, V), \discset, \ell, k, r)$ of {\pvac{{\sav}}} as follows. The candidate set is $C=\xs\cup \{p\}$, and the set of distinguished candidates is $J=\xs$. Regarding the votes, we first create $\frac{3\xsize^2}{4}-3\xsize$ votes, each approving all candidates except~$p$. Let~$V'$ denote the multiset of these votes. Since $\kappa \geq 6$ and is even, $\frac{3\xsize^2}{4}-3\xsize$ is a positive integer. In addition, for each $\xce\in \xc$, we create one vote~$v(\xce)$ which approves exactly the three elements in~$\xce$. Let $V(\xc)=\{v(\xce) \setmid \xce\in \xc\}$. Let $V=V'\cup V(\xc)$. We complete the reduction by setting $k=1$ and $\ell=\xsize$.

It is fairly easy to verify that in~$(C, V)$, the {\sav} score of $p$ is~$0$, and that of each $\xse\in \xs$ is \[\left(\frac{3\xsize^2}{4}-3\xsize\right)\cdot \frac{1}{3\xsize}+1=\frac{\xsize}{4}.\] In the following, we show that~$I$ is a {\yesins} {\iff}~$g(I)$ is a {\yesins}. As~$p$ is the only nondistinguished candidate and $k=1$, excluding any candidate in~$J$ from any winning committees is equivalent to making~$p$ the unique winner.   

$(\Rightarrow)$ Let $\xc'\subseteq \xc$ be an exact $3$-set cover of~$\xs$. Let~$V(\xc')=\{v(\xce) \setmid \xce\in \xc'\}$ be the multiset of votes corresponding to~$\xc'$. Consider the new election obtained from $(C, V)$ by changing each vote in~$V(\xc')$ by the addition of~$p$. In the new election, the {\sav} score of~$p$ increases to $\frac{\xsize}{4}$. As~$\xc'$ is an exact $3$-set cover of~$\xs$, for each $\xse\in \xs$, there is exactly one vote in~$V(\xc')$ which approves~$\xse$. Changing this vote as above decreases the {\sav} score of~$\xse$ by $\frac{1}{3}-\frac{1}{4}=\frac{1}{12}$. Therefore, after all changes, the {\sav} score of every candidate~$\xse\in \xs$ decreases to $\frac{\xsize}{4}-\frac{1}{12}$, implying that~$p$ becomes the unique winner. So,~$g(I)$ is a {\yesins}.
The above correctness argument applies to all $r\geq 1$. 

$(\Leftarrow)$ Assume that~$g(I)$ is a {\yesins}, i.e., we can change a submultiset $\widetilde{V}\subseteq V$ of at most $\ell=\xsize$ votes to make~$p$ the unique winner under the distance bound restriction of~$r$. Notice that, by vote-level addition operation, changing any vote in~$V'$ increases the {\sav} score gap between~$p$ and every other candidate by at most $\frac{1}{3\xsize}$, but changing any vote in~$V(\xc)$ by adding~$p$ increases the {\sav} score gap between~$p$ and any other candidate by at least~$\frac{1}{4}$, which is greater than $\frac{1}{3\xsize}$. 
As a result, we may assume that $\widetilde{V}\subseteq V(\xc)$. Let $\xc'=\{\xce\in \xc \setmid v(\xce)\in \widetilde{V}\}$ be the subcollection of~$\xc$ corresponding to~$\widetilde{V}$. 
Observe further that 
it is always optimal to change the votes in~$\widetilde{V}$ by only adding~$p$ (this holds for all $r\geq 1$). 
Let~$E'$ be the election obtained from~$(C, V)$ by changing every vote in~$\widetilde{V}$ by the addition of~$p$. As discussed above,~$p$ is the unique SAV winner of~$E'$. The {\sav} score of~$p$ in~$E'$ is~$\frac{\xsize}{4}$. It follows that for every candidate $\xse\in \xs$, at least one vote approving~$\xse$ is contained in~$\widetilde{V}$. As a vote~$v(\xce)$ approves a candidate~$\xse$ if and only if $\xse\in \xce$, we know that~$\xc'$ covers~$\xs$. From $\abs{\xc'}=\abs{\widetilde{V}}\leq \xsize$, and the fact that each $\xce\in \xc$ is of cardinality~$3$, we know that every element of~$\xs$ occurs in exactly one element of~$\xc'$, meaning that~$I$ is a {\yesins}. 

The proof for NSAV is a modification of the above reduction based on Lemma~\ref{lem-relation-sav-nsav}.
\end{proof}

\subsection{Vote-Level Deletion Change}
For the vote-level deletion operation, we have an {\nphns} result for {\av} even when we want to elect only one winner.

\begin{theorem}
\label{thm-vdc-av-nph-r-3-k-1}
For all integers $r\geq 3$, {\emph{\pvdc{\av}}} is {\emph\nph}  even when $k=1$.
\end{theorem}

\begin{proof}
We prove the theorem via a reduction from the {\prob{RX3C}} problem. Let~$r$ be an integer at least~$3$. From an instance $I=(\xs,\xc)$ of {\prob{RX3C}} with $\abs{\xs}=\abs{\xc}=3\xsize>0$, we create an instance~$g(I)=((C, V), \discset, \ell, k,r)$ of {\pvdc{\av}} as follows. 
The candidate set is $C=\xs\cup \{p\}$ with $\discset=\xs$ being the set of distinguished candidates. Regarding the votes, we first create three votes, each approving only~$\p$. Then, for each $\xce\in \xc$, we create one vote~$v(\xce)$ approving candidates corresponding to the three elements in~$\xce$. Let $V(\xc)=\{v(\xce) \setmid \xce\in \xc\}$. Let~$V$ be the multiset of the above created $3\xsize+3$ votes. Let $\ell=\xsize$ and~$k=1$. This completes the construction of~$g(I)$. The {\av} scores of all candidates in~$(C, V)$ are~$3$. We show below that~$I$ is a {\yesins} if and only if~$g(I)$ is a {\yesins}.

($\Rightarrow$) Assume that there is an exact set cover $\xc'\subseteq \xc$ of~$\xs$. Let~$E$ be the election obtained from $(C, V)$ by changing the votes corresponding to~$\xc'$ into empty votes. As~$\xc'$ is an exact set cover of~$\xs$, for every candidate $\xse\in \xs$, exactly one vote approving~$\xse$ is changed, leading to the {\av} score of~$\xse$ being $3-1=2$ in~$E$. As none of the votes corresponding to~$\xc$ approves~$p$ originally, the {\av} score of~$p$ remains~$3$ in~$E$. Therefore,~$p$ is the unique winner of $E$, and hence~$g(I)$ is a {\yesins}. 

($\Leftarrow$) Assume that we can change at most $\ell=\xsize$ votes from~$V$ by candidate deletions so that none of~$\xs$ is a winner in the resulting election, or equivalently,~$p$ is the unique winner of the resulting election. 
As all candidates have the same AV score in the original election~$(C, V)$, for every $\xse\in \xs$, at least one vote approving~$\xse$ must be changed to decrease its score (by deleting~$a$ from the vote). By the above construction, only votes from~$V(\xc)$ approve candidates from $\xs$ and, moreover a vote $v(\xce)$ approves a candidate~$\xse\in\xs$ only if $\xse\in \xce$. Hence, the subcollection corresponding to the changed votes from $V(\xc)$ is a set cover of~$\xs$. As at most~$\xsize$ votes are changed, we know that~$I$ is a {\yesins}.
\end{proof}

Next, we show that if every vote is only allowed to delete at most one approved candidate, the problem becomes polynomial-time solvable, regardless of the values of~$k$.

\begin{theorem}
\label{thm-vdc-av-poly-r=1}
{\emph{\pvdc{\av}}}  is polynomial-time solvable if $r=1$.
\end{theorem}

\begin{proof}
We solve the problem by reducing it to the maximum matching problem, which is solvable in polynomial time~\cite{edmonds_1965,DBLP:journals/csur/Galil86}.  

Let $I = (E, \discset, k, \ell, r)$ be an instance of {\pvdc{\av}}, where $E = (C, V)$ and $r = 1$. We first sort all nondistinguished candidates in nonincreasing order of their AV scores in~$E$, breaking ties arbitrarily. Let~$s$ denote the {\av} score of the $k$-th candidate in this order.  

Our goal is to select at most~$\ell$ votes in $V$ and remove one distinguished candidate from each selected vote so that the {\av} score of every distinguished candidate is at most $s - 1$. Let $\discset'$ be the set of distinguished candidates with an {\av} score of at least~$s$ in~$E$.  

If $\sum_{c\in \discset'}(\scoreE{c}{E}{\av} - s + 1) > \ell$,  we immediately conclude that the given instance~$I$ is a {\noins}.

Otherwise, we construct a bipartite graph as follows:  
\begin{itemize}
    \item  For each distinguished candidate in~$\discset'$, we create a vertex.  
\item For each vote that approves at least one candidate from~$\discset'$, we create a vertex.  
\item We introduce an edge between a vote-vertex and a candidate-vertex if the corresponding vote approves the corresponding candidate.   
\end{itemize}

Then, we refine the bipartite graph as follows:  
For each distinguished candidate $c \in \discset'$ with a score of $\scoreE{c}{E}{\av} \geq s+1$, we create $\scoreE{c}{E}{\av} - s$ copies of its corresponding vertex, where each copy retains the same neighbors as the original vertex.  

Finally, we compute a maximum matching of the bipartite graph, and conclude that~$I$ is a {\yesins} if and only if all candidate-vertices and their copies are saturated by the maximum matching.
\end{proof}

Now, we move on to SAV and NSAV.

\begin{theorem}
\label{thm-vdc-sav-nsav-nph-r-3-k-1}
{\emph{\pvdc{{\sav}}}} and {\emph{\pvdc{{\nsav}}}} are {\emph\nph} for all integers $r\geq 3$, even when $k=1$ and every vote approves at most three candidates.
\end{theorem}

\begin{proof}
We give only the proof for {\sav}. The proof for {\nsav} is a slight modification of the following proof based on Lemma~\ref{lem-relation-sav-nsav}.

Our proof is based on a reduction from the {\prob{RX3C}} problem. Let $I=(\xs, \xc)$ be an instance of {\prob{RX3C}}, where $\abs{\xs}=\abs{\xc}=3\xsize>0$. Let~$r$ be an integer such that~$r\geq 3$. We create an instance~$g(I)$ of {\pvdc{{\sav}}} as follows. For each $\xse\in \xs$, we create one candidate, still denoted by~$\xse$ for simplicity. Additionally, we create a set $P=\{p_1,p_2,\dots, p_{r+1}\}$ of~$r+1$ candidates. Let~$C=\xs\cup P$. The votes are structured as follows. We first create one vote $v=P$. Then, for each $\xce\in \xc$, we create one vote~$v(\xce)$ which approves exactly the three candidates in~$\xce$. Let $V(\xc)=\{v(\xce) \setmid \xce\in \xc\}$ be the multiset of votes corresponding to~$\xc$, and let~$V=V(\xc)\cup \{v\}$ denote the multiset of the above created votes. Let $k=1$, $J=\xs$, and $\ell=\xsize+1$. 
The instance $g(I)$ is then $((C, V), J, \ell, k, r)$. We prove below that~$I$ is a {\yesins} if and only if~$g(I)$ is a {\yesins}. Notice that, in the election~$(C, V)$, the {\sav} score of every candidate in~$P$ is $\frac{1}{r+1}\leq \frac{1}{4}$, and that of every $\xse\in \xs$ is~$1$.

$(\Rightarrow)$ Assume that $\xc' \subseteq \xc$ is an exact set cover of~$\xs$. Let $V(\xc')=\{v(\xce) \setmid \xce\in \xc'\}$ be the multiset of votes corresponding to~$\xc'$. We change the~$\xsize$ votes in~$V(\xc')$ into empty votes, and change~$v=P$ into $v'=\{p_1\}$. Clearly, the distance between each changed vote and the new vote is at most~$r$.  As~$\xc'$ is an exact set cover, for every candidate $\xse\in\xs$, there is exactly one~$v(\xce)\in V(\xc')$, $\xse\in \xce\in \xc'$, which approves~$\xse$ originally. Therefore, the above changes decrease the {\sav} score of~$\xse$ by exactly~$\frac{1}{3}$, leading to a final {\sav} score of $1-\frac{1}{3}=\frac{2}{3}$. The {\sav} score of~$p_1$, however, becomes~$1$ after the changes, excluding all distinguished candidates from being a winner.

$(\Leftarrow)$ Assume that we can obtain an election $E'$ by modifying a submultiset $V' \subseteq V$, consisting of at most $\ell = \xsize + 1$ votes, by deleting up to~$r$ candidates from each vote, ensuring that none of the candidates in~$J$ remains a winner~$E'$. 
Obviously, at least one candidate from~$P$ must have a strictly higher {\sav} score than that of any candidate from~$\discset$ in~$E'$. By symmetry, assume that~$p_1\in P$ is such a candidate. We may assume then that $v\in V'$ and~$v$ is changed to approve only~$p_1$ in~$E'$. This assumption is justified because changing~$v$ in this way increases the score gap between~$p_1$ and every distinguished candidate in~$\discset$ by  
$1 - \frac{1}{r+1} \geq \frac{3}{4}$,  
whereas modifying any other vote by merely deleting candidates increases the score gap by at most~$\frac{1}{3}$. Consequently, the {\sav} score of~$p_1$ in~$E'$ is~$1$. This implies that for every $\xse\in \xs$, there is at least one vote approving~$\xse$ that belongs to~$V'$ and is changed. Hence, the subcollection $\xc'=\{\xce\in \xc \setmid v(\xce)\in V'\}$ corresponding to~$V'$ is a set cover of~$\xs$. Since $\abs{\xc'}=\abs{V'\setminus \{v\}}\leq \xsize$, we conclude that~$I$ is a {\yesins}.
\end{proof}

Next, we complement the above theorem with a {\wahns} result. 

\begin{theorem}
\label{thm-vdc-sav-nsav-wah-l-k-r-1}
{\emph{\pvdc{{\sav}}}} and {\emph{\pvdc{{\nsav}}}} for all possible values of $r\geq 1$ are {\emph\wah} with respect to both~$\ell+k$. This holds even when $\abs{\discset}=1$.
\end{theorem}

\begin{proof}
The proof is exactly the one for Theorem~\ref{thm-pvc-sav-nsav-wah-l-k-r-1}. As discussed in the proof of Theorem~\ref{thm-pvc-sav-nsav-wah-l-k-r-1}, when a vote is determined to be changed, an optimal strategy is to delete the dummy candidates and the distinguished candidate from the vote. The rationale for the correctness remains unchanged.
\end{proof}

\section{Fixed-Parameter Tractability
}
\label{sec-many-fpts}
In the preceding sections, we have obtained many intractability results and a few polynomial-time solvability results in several special cases. This section is dedicated to exploring {\fpt}-algorithms concerning three natural parameters: the number of candidates~$m$, the number of voters~$n$, and the number of distinguished candidates~$|J|$. As $|J|\leq m$, any {\fpt}-algorithm with respect to~$|J|$ naturally extends to~$m$.

We have shown that {\pvc{\av}} and {\pvac{\av}} are {\nph}, even when there is only one distinguished candidate (Theorems~\ref{thm-vc-av-wa-hard-k-ell} and~\ref{thm-vac-av-wah}), but left the complexity of {\pvdc{\av}} in this special case unexplored. We address this question now. Particularly, based on integer-linear programming (ILP), we show that the problem is {\fpt} with respect to~$\abs{J}$. This sharply contrasts not only with the above-mentioned hardness of the vote-level operation-based problems for the rule AV, but also with the {\nphns} nature of {\pvdc{\sav}}  and {\pvdc{\nsav}}, even when~$\abs{J}=1$. 

\begin{theorem}
\label{thm-vdc-av-fpt-J}
For all possible values of~$r$, {\emph{\pvdc{\av}}}  is {\emph\fpt} with respect to the number of distinguished candidates.
\end{theorem}

\begin{proof}
Let~$r$ be a nonnegative integer, and let $(E, \discset, k, \ell, r)$ be an instance of {\pvdc{\av}}, where $E=(C, V)$ is an election. 
For each subset $A\subseteq \discset$, let $V(A)=\{v\in V\setmid v\cap \discset=A\}$ denote the multiset of votes approving exactly the candidates in~$A$ among all distinguished candidates, and let $n(A)=\abs{V(A)}$. We compute the {\av} scores of all nondistinguished candidates and rank them in descending order, breaking ties arbitrarily. Let~$s$ denote the {\av} score of the $k$-th candidate in this ranking. For each $A \subseteq \discset$ and each subset $B \subseteq A$ with $\abs{B} \leq r$, we introduce a nonnegative integer variable~$x_{A,B}$, representing the number of votes in~$V(A)$ that are modified by removing only the candidates in~$B$. In total, there are at most $4^{\abs{\discset}}$ such variables. The constraints are as follows.
\begin{itemize}
    \item Since at most~$\ell$ votes can be modified, we enforce
\[\sum_{B\subseteq A, A\subseteq \discset, \abs{B}\leq r}x_{A, B}\leq \ell.\]

\item For each $A\subseteq \discset$, the number of modified votes cannot exceed the number of votes in $V(A)$:
\[\sum_{B\subseteq A, \abs{B}\leq r}x_{A, B}\leq n(A).\]

\item To ensure that no distinguished candidate belongs to any winning $k$-committee, we impose the following constraint for every $c\in \discset$: 
\[\scoreE{c}{E}{\av}-\sum_{c\in B, B\subseteq A, A\subseteq \discset, \abs{B}\leq r}x_{A, B}\leq s-1.\]
\end{itemize}

It is clear that the given instance of {\pvdc{\av}} is a {\yesins} if and only if the above-formulated ILP has a feasible solution.  
By Lenstra's theorem~\cite{DBLP:journals/mor/Lenstra83}, this ILP can be solved in {\fpt}-time with respect to $\abs{\discset}$.  
\end{proof}

\hide{As $\abs{J}$ is smaller than the number of candidates, the following corollary follows.

\begin{corollary}
For all possible values of~$r$, 
{\pvdc{\av}}   is {\fpt} with respect to the number of candidates. 
\end{corollary}
}

Note that the fixed-parameter tractability with respect to~$\abs{J}$ does not extend to {\sav} and {\nsav}. As shown in the previous sections, {\pvdc{{\sav}}} and {\pvdc{{\nsav}}} are {\wah} with respect to~$\ell$ and~$k$, even when $\abs{J}=1$ (Theorem~\ref{thm-vdc-sav-nsav-wah-l-k-r-1}). This distinction arises because, in {\av}, removing a candidate from a vote does not affect the {\av} scores of other candidates in that vote. In contrast, in {\sav} and {\nsav}, removing a candidate from a vote increases the {\sav} scores of the remaining candidates in that vote. This key difference allows us to focus solely on removing distinguished candidates in {\av}, whereas in {\sav} and {\nsav}, all candidates must be considered collectively.

Utilizing ILP formulations once again, we establish that all problems examined in this paper are {\fpt} with respect to a larger parameter, namely, the number of candidates~$m$. 
At a high level, to solve these problems, we first guess the exact winning $k$-committees, each of which excludes all distinguished candidates. There are at most $2^{m\choose k}\leq 2^{2^m}$ guesses, and each guess involves at most~$2^m$ committees. For each guessed class of winning $k$-committees, we provide an ILP formulation. Specifically, we partition the votes into at most~$2^m$ groups, each consisting of all votes approving the same candidates. For each group of votes approving exactly candidates in a subset $A\subseteq C$, we introduce~$2^m$ nonnegative integer variables, with each corresponding to a subset~$B$ of candidates and indicating how many votes from the group are transformed into votes approving exactly the candidates in~$B$. Constraints are derived to ensure that all $k$-committees in the guessed class share the same score, which is strictly higher than that of any $k$-committees not in the class. For vote-level operations with a distance bound~$r$, we impose additional constraints on the variables to ensure that transformations are made only between votes with a Hamming distance at most~$r$. The {\fpt}-running time ensues from the need to solve at most~$2^{2^m}$ ILPs, each of which has at most $2^m\cdot 2^m=4^m$ variables, and ILP is {\fpt} with respect to the number of variables~\cite{DBLP:journals/mor/Lenstra83}.

\begin{theorem}
\label{thm-fpt-m}
For $f\in \{\emph{\text{AV}, \text{{\sav}}, \text{{\nsav}}, \text{{\vccav}}, \text{{\pav}}}\}$, 
{\emph\pappadd{$f$}}, {\emph\pappdel{$f$}}, {\emph\pvc{$f$}}, {\emph\pvac{$f$}}, and {\emph\pvdc{$f$}} are {\emph\fpt} with respect to the number of candidates. The results for the three vote-level operation-based bribery problems hold for all possible values of~$r$.
\end{theorem}

\begin{proof}
Our algorithms for the problems stated in the theorem share a common structure. 
Let $(C, V)$ be an election, $J\subseteq C$ a subset of candidates, and~$k$,~$\ell$, and~$r$ three nonnegative integers. We define~$I$ as  $((C, V), J, k, \ell)$ for an instance of an atomic operation-based problem and as $((C, V), J, k, \ell, r)$ for an instance of a vote-level operation-based problem stated in the theorem. Let $m=\abs{C}$. 

If $\abs{C\setminus J}<k$, we immediately conclude that~$I$ is a {\noins}. Therefore, in what follows, we assume that $\abs{C\setminus J}\geq k$. In addition, if $r=0$, the three vote-level operation-based problems for all the five rules stated in the theorem are clearly {\fpt} with respect to~$m$. Hence, we assume also that $r\geq 1$. 

For $A\subseteq C$, we use~$V(A)$ to denote the multiset of votes in~$V$ approving exactly the candidates in~$A$. 

Our algorithm proceeds by enumerating all nonempty collections of $k$-committees of~$C\setminus \discset$. There are at most $2^{m-\abs{\discset} \choose k}-1$ such collections. For each enumerated collection~$\mathcal{W}$, we formulate an ILP as follows. For each $A, B\subseteq C$, we introduce a nonnegative integer variable denoted by~$x_{A, B}$, representing the number of votes in~$V(A)$ that are modified to approve exactly the candidates in~$B$. There are at most~$4^m$ such variables. The following constraints are imposed.

\begin{itemize}
    \item The first set of constraints are formulated to ensure that the bribery does not surpass the budget~$\ell$: 

\begin{center}
    \begin{tabular}{ll}\toprule
      {\pappadd{$f$}} & $\sum_{A\subsetneq B\subseteq C} \abs{B\setminus A} \cdot x_{A, B}\leq \ell$\\
      {\pappdel{$f$}} & $\sum_{B\subsetneq A\subseteq C} \abs{A\setminus B} \cdot x_{A, B}\leq \ell$\\ 
      {\pvc{$f$}} & $\sum_{A,B\subseteq C, 1\leq \hamdis{A}{B}\leq r}  x_{A, B}\leq \ell$\\
      {\pvac{$f$}} & $\sum_{A\subsetneq B\subseteq C, \hamdis{A}{B}\leq r}  x_{A, B}\leq \ell$\\
      {\pvdc{$f$}}   &  $\sum_{B\subsetneq A\subseteq C, \hamdis{A}{B}\leq r}  x_{A, B}\leq \ell$ \\ \bottomrule
    \end{tabular}
\end{center}

\item For every $A\subseteq C$, we impose the following natural constraints:

\begin{center}
\begin{tabular}{ll}\toprule
      {\pappadd{$f$}} & $\sum_{A\subsetneq B\subseteq C} x_{A, B}\leq \abs{V(A)}$\\
      {\pappdel{$f$}} & $\sum_{B\subsetneq A} x_{A, B}\leq \abs{V(A)}$\\ 
      {\pvc{$f$}} & $\sum_{B\subseteq C, 1\leq \hamdis{A}{B}\leq r}  x_{A, B}\leq \abs{V(A)}$\\
      {\pvac{$f$}} & $\sum_{A\subsetneq B\subseteq C, \hamdis{A}{B}\leq r}  x_{A, B}\leq \abs{V(A)}$\\
      {\pvdc{$f$}}   &  $\sum_{B\subsetneq A, \hamdis{A}{B}\leq r}  x_{A, B}\leq \abs{V(A)}$ \\ \bottomrule
    \end{tabular}
\end{center}

\item We introduce an additional set of constraints to ensure that $\mathcal{W}$ precisely represents the winning $k$-committees after modifying the votes according to the previously defined constraint classes. This is equivalent to guaranteeing that:
\begin{enumerate}
    \item[(1)] All committees in $\mathcal{W}$ have the same score.
    \item[(2)] The score of every committee in $\mathcal{W}$ is strictly higher than that of any committee not in $\mathcal{W}$. 
\end{enumerate} 

This objective can be expressed through ILPs as follows. First, we express the scores of committees in terms of the variables defined above for the rules AV, SAV, NSAV, CCAV, and PAV. 
For each subset $B\subseteq C$, we define $F(B)$ for the five problems as follows:

\begin{center}
\small
{
\begin{tabular}{ll}\toprule
problems & $F(B)$ \\ \midrule
      {\pappadd{$f$}} & $\left(\sum_{A\subsetneq B} x_{A, B}\right)+\left(\abs{V(B)}-\sum_{B\subsetneq A'\subseteq C} x_{B, A'}\right)$\\
      
      {\pappdel{$f$}} &  $\left(\sum_{B\subsetneq A\subseteq C} x_{A, B}\right)+\left(\abs{V(B)}-\sum_{A'\subsetneq B} x_{B, A'}\right)$\\ 
      
      {\pvc{$f$}} & $\left(\sum_{A\subseteq C,1\leq \hamdis{A}{B}\leq r} x_{A, B}\right)+\left(\abs{V(B)}-\sum_{A'\subseteq C, 1\leq \hamdis{A'}{B}\leq r} x_{B, A'}\right)$\\
      
      {\pvac{$f$}} & $\left(\sum_{A\subsetneq B, \hamdis{A}{B}\leq r} x_{A, B}\right)+\left(\abs{V(B)}-\sum_{B\subsetneq A'\subseteq C, \hamdis{A'}{B}\leq r} x_{B, A'}\right)$\\
      
      {\pvdc{$f$}}   &  $\left(\sum_{B\subsetneq A\subseteq C, \hamdis{A}{B}\leq r} x_{A, B}\right)+\left(\abs{V(B)}-\sum_{A'\subsetneq B,\hamdis{A'}{B}\leq r} x_{B, A'}\right)$ \\ \bottomrule
    \end{tabular}
    }
\end{center}

Precisely,~$F(B)$ represents the total number of votes approving exactly the candidates in~$B$ in the final election, after modifying votes in accordance with the previous defined constraints.
For each $w\subseteq C$ and each rule~$f$, we define ${\sf{sc}}(w)$ as follows: 

\begin{center}
\begin{tabular}{ll}\toprule
rules & ${\sf{sc}}(w)$\\ \midrule
AV & $\sum_{B\subseteq C} {\abs{B\cap w}}\cdot F(B)$\\
SAV & $\sum_{B\subseteq C, B\cap w\neq \emptyset} \frac{\abs{B\cap w}}{\abs{B}}\cdot F(B)$\\
NSAV & $\left(\sum_{B\subseteq C, B\cap w\neq \emptyset} \frac{\abs{B\cap w}}{\abs{B}}\cdot F(B)\right)-\left(\sum_{B\subsetneq C} \frac{\abs{B\setminus w}}{m-\abs{B}}\cdot F(B)\right)$\\ 
CCAV & $\sum_{B\subseteq C, B\cap w\neq \emptyset} F(B)$\\
PAV & $\sum_{B\subseteq C, B\cap w\neq \emptyset} (1+1/2+\cdots+1/(\abs{B\cap w}))\cdot F(B)$\\ \bottomrule
\end{tabular}
\end{center}

To achieve the objective described in~(1), we impose the constraint ${\sf{sc}}(w)={\sf{sc}}(w')$ for all $w, w'\in \mathcal{W}$. To achieve the objective described in~(2), we impose the constraint ${\sf{sc}}(w) > {\sf{sc}}(w')$ for every $w\in \mathcal{W}$ and every $k$-committee $w'\subseteq C$ such that $w'\not\in \mathcal{W}$.
\end{itemize}

By Lenstra's theorem~\cite{DBLP:journals/mor/Lenstra83}, each of the above ILPs can be solved in {\fpt}-time with respect to~$m$. Since each problem stated in the theorem requires solving {\fpt}-many ILPs in terms of~$m$, the overall problem can be solved in {\fpt}-time with respect to~$m$.
\end{proof}

%
%
%
%

Finally, we turn our attention to the parameter~$n$, the number of votes. We begin with the following {\fpt}-result.

\begin{theorem}
\label{thm-vdc-av-fpt-n}
For all possible values of~$r$, {\emph\pvdc{\emph{AV}}}  can be solved in time~$\bigos(2^n)$, where~$n$ denotes the number of votes.
\end{theorem}

\begin{proof}
Let $I=(E, \discset, k, \ell, r)$ be an instance of {\pvdc{AV}}, where $E=(C, V)$. Our algorithm proceeds as follows.

We order all nondistinguished candidates from~$C\setminus \discset$  in nonincreasing order based on their AV scores, breaking ties arbitrarily. Let~$s$ denote the score of the $k$-th candidate in the order. Let \[\discset'=\{c\in \discset \setmid \scoreE{c}{E}{\av} \geq s\}\] be the set of distinguished candidates whose AV scores are at least~$s$. If $\discset'=\emptyset$, we directly conclude that~$I$ is a {\yesins}. Otherwise, we need to modify at least one vote to ensure that all distinguished candidates are excluded from all winning $k$-committees. In this case, we split the instance into at most~$2^n$ subinstances, each of which takes as input the original instance along with a nonempty subset $V'\subseteq V$ of at most~$\ell$ votes, and the question is whether we can modify exactly the votes in~$V'$ to exclude all distinguished candidates from any winning $k$-committee. Clearly, there are at most~$2^n$ subinstances, and~$I$ is a {\yesins} if and only if at least one of the subinstances is a {\yesins}. 

To solve a subinstance corresponding to a subset $V'\subseteq V$, we reduce it to a  maximum network flow instance. Particularly, we create a source node~$v^+$ and a sink node~$v^-$. Moreover, for each vote $v\in V'$, we create a node denoted still by~$v$ for simplicity. For each distinguished candidate $c\in \discset'$, we create a node denoted still by~$c$ for simplicity. The arcs in the network are defined as follows:
\begin{itemize}
    \item There is an arc from the source node~$v^+$ to every vote-vertex~$v$, with capacity $\min\{r, \abs{v\cap \discset'}\}$. 
    \item For every vote-node~$v$ and every candidate-node~$c$, there is an arc from~$v$ to~$c$ with capacity~$1$ {\iff}~$v$ approves~$c$, i.e., $c\in v$. 
    \item For each candidate-node $c\in \discset'$, there is an arc from~$c$ to the sink node~$v^-$ with capacity~$\scoreE{c}{E}{\av}-s+1$. 
\end{itemize}
It is not hard to see that the subinstance is a {\yesins}~{\iff}~the above constructed network has a flow of size $\sum_{c\in \discset'} (\scoreE{c}{E}{\av}-s+1)$. The theorem follows from that the maximum network flow problem can be solved in polynomial time (see, e.g.,~\cite{DBLP:conf/stoc/Orlin13}).
\end{proof}

We show that for both the vote-level change operation and the vote-level addition change operation, the corresponding problems remain {\fpt} with respect to~$n$ when $r = m$. To solve these problems, we enumerate all subsets of at most~$\ell$ votes that may be modified. Once the modified votes are identified, we can solve the instance greedily: for the vote-level change operation, we set the modified votes to approve precisely all nondistinguished candidates, while for the vote-level addition change operation, we add all nondistinguished candidates to the modified votes.

\begin{corollary}
\label{thm-vc-av-vac-fpt-n-r-2m}
{\emph\pvc{\emph{AV}}} and {\emph\pvac{\emph{AV}}} for $r=m$ can be solved in time~$\bigos(2^n)$, where~$m$ is the number of candidates and~$n$ is the number of votes.
\end{corollary}

\section{Conclusion}
\label{sec-conclusion}
We have studied the (parameterized) complexity of five destructive bribery problems within the context of approval-based multiwinner voting. These problems simulate scenarios in which a briber seeks to eliminate any possibility of a specified set of candidates winning by bribing voters without exceeding {\their} budget. 
For the five well-studied {\abmrs} AV, {\sav}, {\nsav}, {\vccav}, and {\pav}, we have presented a thorough analysis of the (parameterized) complexity of associated bribery problems. A comprehensive summary of our findings is presented in Table~\ref{tab-results-summary}. 

There are numerous avenues for potential future research. 
\begin{itemize}
    \item To begin, addressing open problems is a natural starting point. While all problems examined in this paper are {\fpt} with respect to the number of candidates, there remains a scarcity of {\fpt}-algorithms relative to the number of votes, leaving numerous cases unresolved. It worth noting that several control and manipulation problems for single-winner voting rules are already known to be {\nph}, even when the number of voters is held  constant~\cite{DBLP:journals/jcss/BachmeierBGHKPS19,DBLP:conf/ijcai/BetzlerNW11,DBLP:journals/jair/ChenFNT17}. 
    Furthermore, for~$\ell$ and $\abs{\discset}$, the parameterized complexity of {\pappadd{SAV}}, {\pappadd{NSAV}}, {\pvac{SAV}}, {\pvac{NSAV}} remains open, and the parameterized complexity of {\pvdc{AV}} for $\ell$ remains open.  
    Additionally, an intriguing parameter that we have not considered is the number of different types of candidates, where two candidates are of the same type if they are approved by exactly the same set of votes. Investigating whether our {\fpt}-algorithms with respect to the number of candidates (Theorems~\ref{thm-fpt-m}) can be extended to this parameter presents a promising direction for future research.   
    
    \item Secondly, exploring the complexity of these problems within special domains of dichotomous preferences presents an open and intriguing area of investigation.  Concepts of a variety of domains can be found in~\cite{DBLP:conf/ijcai/ElkindL15,DBLP:journals/corr/abs-2205-09092,DBLP:journals/aarc/Karpov22,DBLP:conf/ijcai/Yang19a}. It should be pointed out that Kusek~et~al.~\cite{DBLP:conf/atal/KusekBF0K23} have explored the complexity of constructive bribery problems restricted to the candidate interval domain and the voter interval domain. 
    
    \item Thirdly, our study is predominantly theoretical in nature. Conducting experimental work could provide valuable insights into whether these problems pose practical challenges. For promising electoral data available for establishing such works, we recommend consulting~\cite{DBLP:conf/atal/BoehmerS23,DBLP:conf/aldt/MatteiW13}
\end{itemize} 

\section*{Acknowledgement}
The author sincerely thanks the anonymous reviewers of ACM Transactions on Computation Theory  for their insightful and constructive comments, which significantly enhanced the quality of this work. Additionally, the author extends heartfelt gratitude to the anonymous reviewers of AAMAS 2020 for their valuable feedback on an earlier version of this paper.

\end{document}